\documentclass[11pt]{article}
\usepackage[margin=1in]{geometry}
\usepackage{natbib}
\usepackage[dvipsnames]{xcolor}
\usepackage{algorithm}
\usepackage{bm, amssymb, amsmath, amsthm, graphicx, apxproof}

\usepackage{apxproof}

\newcommand{\I}{\mathcal{I}}   
\newcommand{\E}{\mathbb{E}}    
\newcommand{\V}{\mathbb{V}}    
\newcommand{\R}{\mathbb{R}}    
\newcommand{\bO}{\mathcal{O}}  
\newcommand{\Nd}{\mathcal{N}}  
\newcommand{\sv}{\, | \,}      
\newcommand{\eps}{\varepsilon} 
\newcommand{\lam}{\lambda}     
\newcommand{\bx}{\bm{x}}       
\newcommand{\by}{\bm{y}}       
\newcommand{\bv}{\bm{v}}       
\newcommand{\bs}{\bm{s}}       
\newcommand{\bmu}{\bm{\mu}}    
\newcommand{\bS}{\bm{\Sigma}}  
\newcommand{\blam}{\bm{\lam}}  
\newcommand{\dif}{\mathop{}\!\mathrm{d}}   
\newcommand{\norm}[2][]{\left\Vert#2\right\Vert_{#1}} 
\newcommand{\set}[1]{\left\{ #1 \right\}}             
\newcommand{\mat}[1]{\begin{bmatrix}#1\end{bmatrix}}  

\newcommand{\kl}[2]{D_{KL}(#1 \, || \, #2)}

\newtheorem{theorem}{Theorem}

\newtheorem{proposition}{Proposition}

\newcommand{\Dpred}{\mathcal{D}'}
\newcommand{\Spred}{\mathcal{S}'}

\graphicspath{{./figures/}}

\def\spacingset#1{\renewcommand{\baselinestretch}%
{#1}\small\normalsize} \spacingset{1}

\title{Fast and accurate conditioning for large-scale Gaussian process
prediction problems}

\author{Samanyu Arora \and Christopher J. Geoga}

\date{}

\begin{document}

\maketitle

\begin{abstract}
Gaussian Process (GP) models provide a flexible framework for prediction and
uncertainty quantification. For most covariance functions, however, exact GP
prediction with $n$ points scales as $\mathcal{O}(n^3)$, making it prohibitively
expensive for large datasets or large numbers of prediction points. While
nearest neighbor-based prediction can work well in certain settings,
non-pathological circumstances (like measurement noise, for example) can severely
restrict its efficiency. This work presents a complementary approach where one
conditions on carefully designed linear combinations of data, which is
particularly effective in the setting of jointly predicting many values in large
connected regions of the data domain. For kernel functions that are smooth away
from the origin and simple prediction domains, this method can be exponentially
convergent in the number of linear combinations $r$ used for conditioning. The
procedure costs $\mathcal{O}(T r^2)$ work, where $T$ is the cost of solving a
linear system with the data covariance matrix, and so in many cases can be
computed in linear or near-linear cost by exploiting rank structure in
well-behaved covariance matrices. At the cost of $\mathcal{O}(nr^2)$ additional
precomputation work, this approach can also provide predictions at arbitrary
points of a designated region in $\mathcal{O}(1)$ online work, making it
particularly attractive for problems where prediction points are not known in
advance. After establishing favorable theoretical properties, we provide several
example applications to problems in prediction and matrix approximation.

\end{abstract}



\section{Introduction} \label{sec:intro}

Many applications in spatial statistics, computer experiments, and machine
learning require predicting a latent function at a large number of locations
and quantifying uncertainty in those predictions.  Gaussian process (GP)
models are a natural tool for this, as they give conditional expectations that
are linear functions of data and thus conditional distributions that are
themselves multivariate normal distributed. Moreover, they are by far the
easiest process model to specify (at least in dimensions greater than one),
being entirely specified by just their mean and covariance functions. In
particular, a process $f(\bx)$ is a GP with mean function $\mu(\bx) = \E f(\bx)$
and covariance function $K(\bx, \bx') = \text{Cov}(f(\bx), f(\bx'))$ if for any
given collection of locations $\mathcal{S} = \{x_1,\dots,x_n\}$, one has that
the vector $\bm{f} \in \R^n$ with $\bm{f}_j = f(\bx_j)$ has distribution 
\begin{equation*} 
  \bm{f} \sim \mathcal{N}(\bmu, \bS), 
  \quad \text{with} \quad 
  \bmu_j = \mu(\bx_j), \, \bS_{j,k} = K(\bx_j, \bx_k).
\end{equation*}
In many applications, one presumes to make measurements that are also polluted
with independent noise, so that actual measurements at a location $\bx$, denoted
$y(\bx)$, are given by the observation model
\begin{equation*}
\label{eq:obs_model}
y_i = f(\bx_i) + \varepsilon_i,
\qquad 
\boldsymbol{\varepsilon}\sim \mathcal{N}(\mathbf{0},\tau^2 \mathbf{I}_n).
\end{equation*}
We will make two simplifying assumptions in this work. First, we will assume
that $\mu(\bx) = \E f(\bx) \equiv 0$, so that $f$ (and thus $y$) are
\emph{mean-zero}, which is a common and mild assumption in the literature since
means are often estimated and handled separately. Second, we will assume that
the covariance function $K(\bx, \bx')$ is \emph{translation invariant}, so that
$K(\bx, \bx') = K(\bx - \bx')$. This makes the process $f(\bx)$
\emph{stationary}. This assumption is also very common among practitioners,
although it is not necessarily mild. Much of the analysis and methodology in
this work could readily be extended to the nonstationary case at little
conceptual but significant technical expense.

For the problem of predicting $f(\bx')$ at prediction locations $\mathcal{S}' =
\{\bx'_1,\dots,\bx'_m\}$, the conditional distribution of $\bm{f} =
[f(\bx'_k)]_{k=1}^m$ given $\by$ is given by
\begin{align} \label{eq:conditional}
\bm{f} \sv \by
\sim
\mathcal{N}\left(
\bm{C}^T
\bS^{-1}
\mathbf{y},
\;
\bm{K}(\mathcal{S}', \mathcal{S}')
-
\bm{C}^T
\bS^{-1}
\bm{C}
\right)
=
\mathcal{N}\left(
E[\bm{f} \sv \by],
\;
V[\bm{f} \sv \by]
\right),
\end{align}
where $\bm{C} = K(\mathcal{S}, \mathcal{S}') = [K(\bx_j,
\bx'_k)]_{j=1,k=1}^{n,m}$ denotes the kernel-implied cross covariance matrix and
$\bS = (K(\mathcal{S}, \mathcal{S}) + \tau^2 \I)$ is the marginal variance for
observed data $\by$.  As this fact demonstrates, the covariance function that
dictates the structure of matrices like $\bm{C}$ and $\bS$ is extremely
important in dictating the predictive properties of a GP model, particularly in
the conditional variances which do not depend on the data at all and are
\emph{entirely} determined by the choice of covariance function.

While its simplicity and concrete interpretation is appealing, computing these
conditional distributions can be expensive: both the conditional mean and
conditional variance require factorizing $\bS$, which costs $\bO(n^3)$ work and
$\bO(n^2)$ storage. Given that many modern GP applications are done with large
datasets, this is often prohibitively expensive to do exactly. And even in
settings where one is able to compute a conditional mean exactly, for the
problem of predicting at \emph{many} locations, or rapidly predicting in an
online way as one might do in Bayesian optimization that involves optimizing
$\bx' \mapsto E[f(\bx') \sv \by] + \sqrt{V[f(\bx') \sv \by]}$ \cite{mockus2012},
the computational burden grows significantly more expensive still. For this
reason, the vast majority of applications that require these variety of
predictive tasks use approximate methods.

A theme connecting many of the most performant scalable Gaussian process
approximations of today is \emph{locality}. The simple and natural idea that the
prediction $E[f(\bx_0) \sv \by]$ depends primarily on measurements in $\by$ made
closest to $\bx_0$ motivates the most broadly applicable and successful methods
for approximating Gaussian log-likelihoods in the setting of infeasibly large
data sizes. Methods under this umbrella go by several names, including
\emph{Vecchia} approximations
\cite{vecchia1988,stein2004,katzfuss2021,schafer2021}, nearest neighbor Gaussian
processes \cite{datta2016,finley2019}, Gauss-Markov random field (GMRF) models
\cite{rue2005,lindgren2011}, and
other popular methods such as inducing points \cite{snelson2007} can be interpreted as
special cases of this broader phenomenon \cite{kaminetz2026}. The thread
connecting all of these ideas is that the \emph{conditional} distribution
$f(\bx') \sv \by$ can in certain circumstances be well-approximated with
$f(\bx') \sv \by_{\sigma(\bx')}$, where $\sigma(\bx')$ denotes a small subset of
the full data typically consisting of the most nearby measurements to $f(\bx')$.
Approximating conditionals in this way offers clear computational improvements,
reducing the computational burden of the log-likelhood from $\bO(n^3)$ to
$\bO(k^3)$ if each $\sigma(\bx_j)$ is $\bO(1)$ in size. A fundamental limitation
for these subsetting-flavored prediction methods, however, is that their
accuracy depends on the covariance model satisfying the \emph{screening}
property \cite{stein2011}, which is closely connected to having Markovian
structure.  Unfortunately, Markovian-like structure can be brittle, and very
common modeling assumptions and designs---such as having measurement error---can
significantly impact the accuracy of these approximations \cite{stein2011}.
Methods to compensate for this degradation have been proposed and demonstrated
to work well \cite{schafer2021,katzfuss2021,geoga2024}. The mechanism for these
corrections, however, is that they work with approximate conditionals in a
latent space where locality is restored. As such, for the direct problem of
approximating conditional means and variances, they are not applicable.

\begin{figure}
  \centering
\begingroup
  \makeatletter
  \providecommand\color[2][]{%
    \GenericError{(gnuplot) \space\space\space\@spaces}{%
      Package color not loaded in conjunction with
      terminal option `colourtext'%
    }{See the gnuplot documentation for explanation.%
    }{Either use 'blacktext' in gnuplot or load the package
      color.sty in LaTeX.}%
    \renewcommand\color[2][]{}%
  }%
  \providecommand\includegraphics[2][]{%
    \GenericError{(gnuplot) \space\space\space\@spaces}{%
      Package graphicx or graphics not loaded%
    }{See the gnuplot documentation for explanation.%
    }{The gnuplot epslatex terminal needs graphicx.sty or graphics.sty.}%
    \renewcommand\includegraphics[2][]{}%
  }%
  \providecommand\rotatebox[2]{#2}%
  \@ifundefined{ifGPcolor}{%
    \newif\ifGPcolor
    \GPcolortrue
  }{}%
  \@ifundefined{ifGPblacktext}{%
    \newif\ifGPblacktext
    \GPblacktexttrue
  }{}%
  \let\gplgaddtomacro\g@addto@macro
  \gdef\gplbacktext{}%
  \gdef\gplfronttext{}%
  \makeatother
  \ifGPblacktext
    \def\colorrgb#1{}%
    \def\colorgray#1{}%
  \else
    \ifGPcolor
      \def\colorrgb#1{\color[rgb]{#1}}%
      \def\colorgray#1{\color[gray]{#1}}%
      \expandafter\def\csname LTw\endcsname{\color{white}}%
      \expandafter\def\csname LTb\endcsname{\color{black}}%
      \expandafter\def\csname LTa\endcsname{\color{black}}%
      \expandafter\def\csname LT0\endcsname{\color[rgb]{1,0,0}}%
      \expandafter\def\csname LT1\endcsname{\color[rgb]{0,1,0}}%
      \expandafter\def\csname LT2\endcsname{\color[rgb]{0,0,1}}%
      \expandafter\def\csname LT3\endcsname{\color[rgb]{1,0,1}}%
      \expandafter\def\csname LT4\endcsname{\color[rgb]{0,1,1}}%
      \expandafter\def\csname LT5\endcsname{\color[rgb]{1,1,0}}%
      \expandafter\def\csname LT6\endcsname{\color[rgb]{0,0,0}}%
      \expandafter\def\csname LT7\endcsname{\color[rgb]{1,0.3,0}}%
      \expandafter\def\csname LT8\endcsname{\color[rgb]{0.5,0.5,0.5}}%
    \else
      \def\colorrgb#1{\color{black}}%
      \def\colorgray#1{\color[gray]{#1}}%
      \expandafter\def\csname LTw\endcsname{\color{white}}%
      \expandafter\def\csname LTb\endcsname{\color{black}}%
      \expandafter\def\csname LTa\endcsname{\color{black}}%
      \expandafter\def\csname LT0\endcsname{\color{black}}%
      \expandafter\def\csname LT1\endcsname{\color{black}}%
      \expandafter\def\csname LT2\endcsname{\color{black}}%
      \expandafter\def\csname LT3\endcsname{\color{black}}%
      \expandafter\def\csname LT4\endcsname{\color{black}}%
      \expandafter\def\csname LT5\endcsname{\color{black}}%
      \expandafter\def\csname LT6\endcsname{\color{black}}%
      \expandafter\def\csname LT7\endcsname{\color{black}}%
      \expandafter\def\csname LT8\endcsname{\color{black}}%
    \fi
  \fi
    \setlength{\unitlength}{0.0500bp}%
    \ifx\gptboxheight\undefined%
      \newlength{\gptboxheight}%
      \newlength{\gptboxwidth}%
      \newsavebox{\gptboxtext}%
    \fi%
    \setlength{\fboxrule}{0.5pt}%
    \setlength{\fboxsep}{1pt}%
    \definecolor{tbcol}{rgb}{1,1,1}%
\begin{picture}(7920.00,3400.00)%
    \gplgaddtomacro\gplbacktext{%
    }%
    \gplgaddtomacro\gplfronttext{%
      \csname LTb\endcsname
      \put(1036,3112){\makebox(0,0){\strut{}\small Prediction domain}}%
    }%
    \gplgaddtomacro\gplbacktext{%
      \csname LTb\endcsname
      \put(3355,676){\makebox(0,0)[r]{\strut{}\footnotesize $10^{-7}$}}%
      \csname LTb\endcsname
      \put(3355,989){\makebox(0,0)[r]{\strut{}\footnotesize $10^{-6}$}}%
      \csname LTb\endcsname
      \put(3355,1303){\makebox(0,0)[r]{\strut{}\footnotesize $10^{-5}$}}%
      \csname LTb\endcsname
      \put(3355,1617){\makebox(0,0)[r]{\strut{}\footnotesize $10^{-4}$}}%
      \csname LTb\endcsname
      \put(3355,1931){\makebox(0,0)[r]{\strut{}\footnotesize $10^{-3}$}}%
      \csname LTb\endcsname
      \put(3355,2245){\makebox(0,0)[r]{\strut{}\footnotesize $10^{-2}$}}%
      \csname LTb\endcsname
      \put(3355,2559){\makebox(0,0)[r]{\strut{}\footnotesize $10^{-1}$}}%
      \csname LTb\endcsname
      \put(3355,2872){\makebox(0,0)[r]{\strut{}\footnotesize $10^{0}$}}%
      \csname LTb\endcsname
      \put(3456,436){\makebox(0,0){\strut{}\footnotesize 5}}%
      \csname LTb\endcsname
      \put(3892,436){\makebox(0,0){\strut{}\footnotesize 25}}%
      \csname LTb\endcsname
      \put(4438,436){\makebox(0,0){\strut{}\footnotesize 50}}%
      \csname LTb\endcsname
      \put(5529,436){\makebox(0,0){\strut{}\footnotesize 100}}%
    }%
    \gplgaddtomacro\gplfronttext{%
      \csname LTb\endcsname
      \put(7401,2753){\makebox(0,0)[r]{\strut{}\footnotesize KNN, $\tau^2 = 0.00$}}%
      \csname LTb\endcsname
      \put(7401,2513){\makebox(0,0)[r]{\strut{}\footnotesize KNN, $\tau^2 = 0.01$}}%
      \csname LTb\endcsname
      \put(7401,2273){\makebox(0,0)[r]{\strut{}\footnotesize KNN, $\tau^2 = 0.10$}}%
      \csname LTb\endcsname
      \put(7401,2033){\makebox(0,0)[r]{\strut{}\footnotesize KNN, $\tau^2 = 0.25$}}%
      \csname LTb\endcsname
      \put(7401,1794){\makebox(0,0)[r]{\strut{}\footnotesize $\bm{s} = \bm{Q}^T \by$, $\tau^2 = 0.00$}}%
      \csname LTb\endcsname
      \put(7401,1554){\makebox(0,0)[r]{\strut{}\footnotesize $\bm{s} = \bm{Q}^T \by$, $\tau^2 = 0.01$}}%
      \csname LTb\endcsname
      \put(7401,1314){\makebox(0,0)[r]{\strut{}\footnotesize $\bm{s} = \bm{Q}^T \by$, $\tau^2 = 0.10$}}%
      \csname LTb\endcsname
      \put(7401,1075){\makebox(0,0)[r]{\strut{}\footnotesize $\bm{s} = \bm{Q}^T \by$, $\tau^2 = 0.25$}}%
      \csname LTb\endcsname
      \csname LTb\endcsname
      \put(4493,76){\makebox(0,0){\strut{}neighbors $k$ or size $\bm{s} \in \R^r$}}%
      \csname LTb\endcsname
      \put(4493,3112){\makebox(0,0){\strut{}\small $f(\bx_0) \sv \by$ approximation error}}%
    }%
    \gplbacktext
    \put(0,0){\includegraphics[width={396.00bp},height={170.00bp}]{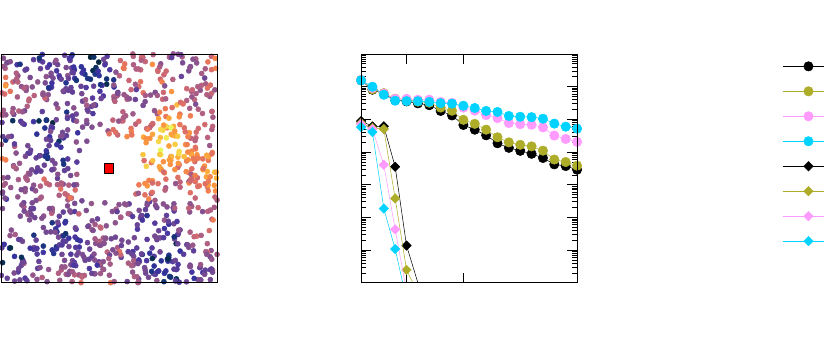}}%
    \gplfronttext
  \end{picture}%
\endgroup

  \caption{ 
    An accuracy comparison for the problem of approximating $f(\bx_0) \sv \by$
    with $\bx_0$ at the center of the square $[0,1]^2$ (example data and domain
    shown in the left panel) with a Mat\'ern$(\nu = 0.75)$ model and additive
    noise variance $\tau^2$ using either $k$ nearest neighbors or linear
    combinations $\bm{s} = \bm{Q}^T \by \in \R^r$ (the new approach of this
    work) for conditioning. The error metric used here (Kullback-Leibler
    divergence) is a closeness measure for probability distributions and is
    defined in Section
    \ref{sec:low-rank-crosscov}.
  }
  \label{fig:intro_knn_compare}
\end{figure}

In this work, we offer an alternative---and in many ways
complementary---approach to rapid and accurate prediction, including
$\bO(1)$-cost online prediction per location, in a given domain. Rather than
reducing the cost of prediction by conditioning on only a small subset of data,
we condition on a small number of \emph{dense} linear combinations of data,
denoted $\bm{s} = \bm{Q}^T \by$ for orthogonal matrix $\bm{Q} \in \R^{n \times
r}$, in a way that somewhat resembles a Vecchia-type approximation but for a
carefully transformed version of the data. As will be demonstrated in the next
section, if the covariance function being used satisfies mild conditions such as
smoothness away from the origin, prediction problems on large connected portions
of a domain can often be made effectively exact by conditioning on $\sim 100$ or
fewer carefully chosen linear combinations, and with additional precomputations
can provide online evaluation of effectively exact conditional
distributions---including conditional variances---in $\bO(1)$ work after an
additional $\bO(n)$-cost precomputation.  Figure \ref{fig:intro_knn_compare}
gives a demonstration of $k$-nearest neighbor prediction versus the the approach
described here at various levels of measurement noise with a Mat\'ern covariance
function.

While the analysis of this strategy from a probabilistic perspective is novel to
our knowledge, it is worth noting that many of the analytical techniques used in
this work are directly borrowed from a significant existing body of literature
in the fast algorithm literature. The notions of exploiting low-rank structure
in cross-covariances between well separated regions, for example, dates back at
least to the fast multipole method (FMM) \cite{greengard1987}, and the notation
and terminology around kernel interpolation and anterpolation closely follows
those conventions. Further, we note that several existing methods in the GP
literature have been proposed that address problems like scalably predicting at
significant numbers of spatial locations \cite{chen2023, greengard2022}.
Similarly, alternative methods for approximating and accelerating conditional
(co)variances have recently been presented in \cite{cai2025}. The primary
contribution of this work is to provide an analysis of this prediction mechanism
from the perspective of statistical efficiency, and to describe a mechanism by
which a practitioner with access to rapid linear solves $\bv \mapsto \bS^{-1}
\bv$ can build high-efficiency approximations to conditional distributions,
both in the numerical and statistical sense.

The remainder of the paper is as follows. In section \ref{sec:conditioning}, we
will outline the probabilistic mechanism being exploited that can make this
low-dimensional conditional distribution effectively as efficient as its generic
full-dimensional counterpart. Following that, in Section \ref{sec:computation}
we discuss a broadly applicable computational apparatus for obtaining these
special linear combinations and the requisite matrices for conditional
(co)-variances in a time complexity that is linear in the cost of a linear solve
$\bv \mapsto \bS^{-1} \bv$, which often coincides with the cost of $\bv \mapsto
\bS \bv$ due to the efficiency of common preconditioners for well-behaved
covariance matrices like the Vecchia approximation \cite{vecchia1988} (see
\cite{guinness2021} for an example, and \cite{geoga_vecchia_jl} for a performant
software library and runnable demonstrations).  Finally, Section \ref{sec:demo}
will provide numerical demonstrations showing the accuracy and efficiency of our
method in a variety of applications and tasks.

\section{Conditioning on Linear Combinations} \label{sec:conditioning}
A key observation that this section will outline in detail is that there is a
fundamental connection between smoothness of the covariance function away from
the origin and a reducible form for shared information between measurements in
two well-separated regions of a spatial domain. As mentioned in the
introduction, the deterministic aspect of this observation underlies many
classical fast algorithms for matrix-vector products
\cite{greengard1987,toukmaji1996,jiang2025}. What is new in this discussion,
however, is the concrete probabilistic analysis of how this structure can be
used to create low-dimensional transformations of the data $\by$ resembling
sufficient statistics.

Throughout this section unless otherwise stated, we will use $\mathcal{S} =
\{\bx_1,\dots,\bx_n\}$ to denote the training inputs and $\by \in \R^n$ for the
noisy observations (so that $\by_j = f(\bx_j) + \eps_j$). Moreover, we will
consistently use $\mathcal{S}' = \set{\bx'_j}_{j=1}^m \subset \Dpred$ to denote
an arbitrary collection of points at which we wish to predict $\bm{f} =
[f(\bx'_j)]_{j=1}^m$. We will assume that the nice regions containing the data
locations $\mathcal{S}$, denoted $\mathcal{D}$, and the region containing
prediction locations $\mathcal{S}'$, denoted $\Dpred$, are positively separated
under any reasonable distance metric. Finally, we will use $\bm{s} = \bm{Q}^T
\by$ to denote the lower-dimensional set of linear combinations of data that we
will use to approximate $\bm{f} \sv \by$. Proofs of all results are provided in
the appendix.

\subsection{Probabilistic idea: conditioning on linear summaries}
\label{sec:prob-idea}

We begin with the joint distribution of the observed data and the latent
function value at a single prediction location. Let $f_0 = f(\bx_0)$ be the
latent value of the process at a new location $\bx_0$. Under a Gaussian process
prior with covariance $K$ and noise variance $\tau^2$, Equation
\ref{eq:conditional} simplifies to
\begin{equation}\label{eq:condDist}
f_0 \sv \by
\sim
\mathcal{N}\!\bigl(
\bm{c}_0^{\!T}\bm{\Sigma}^{-1} \by,\;
K(\bx_0, \bx_0) - \bm{c}_0^{\!T}\bm{\Sigma}^{-1} \bm{c}_0
\bigr),
\qquad
\bm{c}_0 = K(\mathcal{S},\bx_0).
\end{equation}
A basic definitional fact about conditional expectations is that for any two random
variables $W_1$ and $W_2$, $E[W_1 \sv W_2] = E[W_1 \sv E[W_1 \sv W_2]]$. In the
Gaussian process setting, this general result can be interpreted and proven in
particularly concrete terms that motivate the more general statements that will
be given shortly.

\begin{proposition}
\label{thm:kriging-sufficiency}
Define $\blam = \bm{\Sigma}^{-1} \bm{c}_0 \in \R^n$ and $s = \blam^T \by$.  Then
$f_0 \sv s \;\stackrel{d}{=}\; f_0 \sv \by$, so that conditioning on the scalar
$s$ yields the same conditional distribution for $f_0$ as conditioning on the
full data $\by$.
\end{proposition}
\begin{toappendix}
\begin{proof}[Proof of Proposition \ref{thm:kriging-sufficiency}]
Noting that
$
\text{Var}(s)
=
\blam^T \bm{\Sigma} \blam
= 
\bm{c}_0^{\!T}\bm{\Sigma}^{-1} \bm{c}_0,
$
a standard computation gives the conditional mean $E[f_0 \sv s]$ as
\begin{equation*}
E[f_0 \sv s]
= \bm{c}_0^{\!T}\blam\;(\blam^T \bm{\Sigma} \blam)^{-1}\;\blam^T \by
= \bm{c}_0^{\!T}\bm{\Sigma}^{-1} \bm{c}_0\;
\bigl(\bm{c}_0^{\!T}\bm{\Sigma}^{-1} \bm{c}_0\bigr)^{-1}\;
\bm{c}_0^{\!T}\bm{\Sigma}^{-1} \by
= \bm{c}_0^{\!T}\bm{\Sigma}^{-1} \by.
\end{equation*}
A similar computation gives the conditional variance $V[f_0 \sv s]$ as
\begin{equation*}
V[f_0 \sv s]
= K(\bx_0, \bx_0) - \bm{c}_0^{\!T}\blam\;(\blam^T \bm{\Sigma}
\blam)^{-1}\;\blam^T \bm{c}_0
= K(\bx_0, \bx_0) - \bm{c}_0^{\!T}\bm{\Sigma}^{-1} \bm{c}_0.
\end{equation*}
These coincide with the mean and variance of $f_0 \sv \by$, and so since the
multivariate normal is fully specified by its first two moments we conclude the
two distributions are identical.
\end{proof} 
\end{toappendix}
Naturally, this is not particularly useful at face value since obtaining $E[f_0
\sv \by] = \blam^T \by$ and then conditioning on that scalar does not actually
save the user any computation. But it motivates the \emph{idea} that there are
specific linear combinations of the data that may contain nearly all of the
predictive information of a given data vector $\by$ as it pertains to an
unknown value $f(\bx_0)$. And if linear predictor vectors $\blam$ for many
prediction points live in a low-dimensional subspace, then a reasonably small
number of linear combinations may be strongly explanatory for a large number of
prediction locations. Letting $\bm{\alpha} = \bm{\Sigma}^{-1} \by$, we note that the Kriging estimator
for $f(\bx_0)$ can be written in the above way or in the potentially more
familiar kernel interpolation form as
\begin{equation} \label{eq:krige-dual}
  E[f(\bx_0) \sv \by] = \blam(\bx_0)^T \by = \sum_{j=1}^n \alpha_j K(\bx_j-\bx_0).
\end{equation}
And so particularly in settings where one wants to predict at many $\bx'$ values
in some connected region of space that contains no or a small number of data
points which we will denote with $\Dpred$, then smoothness of $K$ in the
appropriate regions of its domain will impart smoothness to the predictive mean
$\hat{f}$ on $\Dpred$. The crux of our approach, as we will now demonstrate, is
that this smoothness is exploitable to obtain such special linear combinations
that are informative and whose corresponding matrix operations can be
accelerated.

\subsection{Low-rank structure in cross-covariance matrices}
\label{sec:low-rank-crosscov}

A very common and desirable property of covariance functions is that they are
smooth away from the origin. This smoothness means that if $\bx \in
\mathcal{D}$, for a suitable family of basis functions $\set{T_j}$ the
approximation
\begin{equation*} 
  \bx' \mapsto K(\bx - \bx') \approx 
  \sum_{j=1}^R \beta_jT_j (\bx'),
  \qquad 
  \bx' \in \Dpred,
\end{equation*}
converges rapidly with respect to $R$. In the following section, we will develop
an error analysis for joint conditional distributions of prediction points in
$\Dpred$ that presumes such a suitable basis has been identified, and that
\begin{equation} \label{eq:inf_error}
  \eps_R(f, \hat{f}) =
  \sum_{j > R}\bigl|\beta_{j}\bigr|,
  \quad
  \text{with}
  \quad
  f = \sum_{j=1}^\infty \beta_{j} T_{j},
  \quad
  \hat{f} = \sum_{j=1}^R \beta_{j} T_{j},
\end{equation}
an upper bound for $\norm[{L^{\infty}(\Dpred)}]{f - \hat{f}}$, converges rapidly
with $R$. For simple domains, this problem has been extremely well-studied and
excellent options are known to exist. If $\Dpred$ is a hyper-rectangle, for
example, tensor-products of Chebyshev polynomials are a natural choice for
$\set{T_j}$, and if $f$ is smooth then $\eps_R$ will decay exponentially with
$R$ \cite{trefethen2019}. For approximation on a disk in two dimensions, Zernike
polynomials are the analogous natural choice \cite{niu2022}. For the exposition
of this work, we will presume that the function $\bx' \mapsto K(\bx - \bx')$
will be approximated on $\Dpred$, presumed to be a hyperrectangle, so that
$\set{T_j}$ will be chosen as a tensor product of Chebyshev polynomials. But we
stress that the actual choice of basis $\set{T_j}$ itself is not critical to the
below analysis---all that matters is that there exists \emph{some} basis for
which the above error metric $\eps_R$ converges rapidly with $R$. For
convenience, we will assume that each $T_j$ is uniformly bounded by one. Let
$\{\tilde{\bx}_i\}_{i=1}^R \subset \Dpred$ be a collection of Chebyshev nodes.
Following the notation of \cite{jiang2025}, define
\begin{equation*}
\bm{V}\in\mathbb R^{R\times R},\quad \bm{V}_{i,j}=T_j(\tilde{\bx}_i),\qquad
\bm{E}\in\mathbb R^{m\times R},\quad \bm{E}_{k,j}=T_j(\bx'_k),\qquad
\bm{U}=\bm{E}\bm{V}^{-1}\in\mathbb R^{m\times R},
\end{equation*}
so that $\bm{U}$ is the interpolation operator on $\Dpred$. For each fixed
$\bx\in \mathcal{S}$, smoothness of $\bx'\mapsto K(\bx-\bx')$ on $\Dpred$ yields
\begin{equation} \label{eq:decomp}
K(\bx-\bx'_k)\ \approx\ \sum_{i=1}^R \bm{U}_{k,i}\,K(\bx-\tilde{\bx}_i),\qquad k=1,\ldots,m.
\end{equation}
Consequently, for any $\bm{\rho}\in\mathbb R^m$,
\begin{equation*}
\sum_{k=1}^m \rho_k\,K(\bx-\bx'_k)\ \approx\ \sum_{i=1}^R \tilde\rho_i\,K(\bx-\tilde{\bx}_i),
\qquad \tilde{\bm{\rho}}=\bm{U}^T \bm{\rho},
\end{equation*}
and we refer to $\bm{U}^T$ as the \emph{anterpolation} matrix. Writing
$\tilde{\bm{C}}_R =K(\mathcal{S}, \set{\tilde{\bx}_j}_{j=1}^R)$, this
implies the cross-covariance compression
\begin{equation} \label{eq:C_approx}
\bm{C}=\bigl[K(\bx_i-\bx'_k)\bigr]_{i=1,k=1}^{n,m}\ \approx\
\tilde{\bm{C}}_R \,\bm{U}^T
\end{equation}
\emph{for any collection of points $\set{\bx'_k} \subset \Dpred$}, providing a
formal explanation for low rank structure in the cross covariance matrix
$\bm{C}$. In what follows, we will describe an error analysis procedure based on
extracting the non-degenerate column basis (or dominant singular vectors of)
$\tilde{\bm{\Lambda}}_R = \bm{\Sigma}^{-1} \tilde{\bm{C}}_R \approx \bm{Q} \bm{R}$, where
$\bm{Q} \in \R^{n \times r}$ is an orthogonal matrix representing that column
space basis. $\tilde{\bm{\Lambda}}_R$ is interpretable as the matrix whose columns
give Kriging weights for the interpolation points distributed in
$\Dpred$. As we will demonstrate, because of the expansion (\ref{eq:decomp}),
this low-dimensional subspace will \emph{almost} contain Kriging weights
$\blam(\bx_0) = \bm{\Sigma}^{-1} K(\mathcal{S}, \bx_0)$ for \emph{any} $\bx_0 \in
\Dpred$. And since the distribution $f(\bx_0) \sv \by$ is precisely equal to
$f(\bx_0) \sv \blam(\bx_0)^T \by$ by Proposition \ref{thm:kriging-sufficiency},
we will show that $f(\bx_0) \sv \bm{Q}^T \by$ with $\bm{Q}$ as above is an
excellent approximation to $f(\bx_0) \sv \by$ for any $\bx_0 \in \Dpred$.
Example visualizations of the columns of $\bm{Q}$ for a Gaussian
covariance function are given in Figure \ref{fig:Q_demo}, demonstrating that
these dominant singular vectors of $\tilde{\bm{\Lambda}}_R$ are very non-local
(and thus certainly not captured by $k$-nearest neighbor conditioning, for
example).

\begin{figure}
  \centering
\begingroup
  \makeatletter
  \providecommand\color[2][]{%
    \GenericError{(gnuplot) \space\space\space\@spaces}{%
      Package color not loaded in conjunction with
      terminal option `colourtext'%
    }{See the gnuplot documentation for explanation.%
    }{Either use 'blacktext' in gnuplot or load the package
      color.sty in LaTeX.}%
    \renewcommand\color[2][]{}%
  }%
  \providecommand\includegraphics[2][]{%
    \GenericError{(gnuplot) \space\space\space\@spaces}{%
      Package graphicx or graphics not loaded%
    }{See the gnuplot documentation for explanation.%
    }{The gnuplot epslatex terminal needs graphicx.sty or graphics.sty.}%
    \renewcommand\includegraphics[2][]{}%
  }%
  \providecommand\rotatebox[2]{#2}%
  \@ifundefined{ifGPcolor}{%
    \newif\ifGPcolor
    \GPcolortrue
  }{}%
  \@ifundefined{ifGPblacktext}{%
    \newif\ifGPblacktext
    \GPblacktexttrue
  }{}%
  \let\gplgaddtomacro\g@addto@macro
  \gdef\gplbacktext{}%
  \gdef\gplfronttext{}%
  \makeatother
  \ifGPblacktext
    \def\colorrgb#1{}%
    \def\colorgray#1{}%
  \else
    \ifGPcolor
      \def\colorrgb#1{\color[rgb]{#1}}%
      \def\colorgray#1{\color[gray]{#1}}%
      \expandafter\def\csname LTw\endcsname{\color{white}}%
      \expandafter\def\csname LTb\endcsname{\color{black}}%
      \expandafter\def\csname LTa\endcsname{\color{black}}%
      \expandafter\def\csname LT0\endcsname{\color[rgb]{1,0,0}}%
      \expandafter\def\csname LT1\endcsname{\color[rgb]{0,1,0}}%
      \expandafter\def\csname LT2\endcsname{\color[rgb]{0,0,1}}%
      \expandafter\def\csname LT3\endcsname{\color[rgb]{1,0,1}}%
      \expandafter\def\csname LT4\endcsname{\color[rgb]{0,1,1}}%
      \expandafter\def\csname LT5\endcsname{\color[rgb]{1,1,0}}%
      \expandafter\def\csname LT6\endcsname{\color[rgb]{0,0,0}}%
      \expandafter\def\csname LT7\endcsname{\color[rgb]{1,0.3,0}}%
      \expandafter\def\csname LT8\endcsname{\color[rgb]{0.5,0.5,0.5}}%
    \else
      \def\colorrgb#1{\color{black}}%
      \def\colorgray#1{\color[gray]{#1}}%
      \expandafter\def\csname LTw\endcsname{\color{white}}%
      \expandafter\def\csname LTb\endcsname{\color{black}}%
      \expandafter\def\csname LTa\endcsname{\color{black}}%
      \expandafter\def\csname LT0\endcsname{\color{black}}%
      \expandafter\def\csname LT1\endcsname{\color{black}}%
      \expandafter\def\csname LT2\endcsname{\color{black}}%
      \expandafter\def\csname LT3\endcsname{\color{black}}%
      \expandafter\def\csname LT4\endcsname{\color{black}}%
      \expandafter\def\csname LT5\endcsname{\color{black}}%
      \expandafter\def\csname LT6\endcsname{\color{black}}%
      \expandafter\def\csname LT7\endcsname{\color{black}}%
      \expandafter\def\csname LT8\endcsname{\color{black}}%
    \fi
  \fi
    \setlength{\unitlength}{0.0500bp}%
    \ifx\gptboxheight\undefined%
      \newlength{\gptboxheight}%
      \newlength{\gptboxwidth}%
      \newsavebox{\gptboxtext}%
    \fi%
    \setlength{\fboxrule}{0.5pt}%
    \setlength{\fboxsep}{1pt}%
    \definecolor{tbcol}{rgb}{1,1,1}%
\begin{picture}(7920.00,3400.00)%
    \gplgaddtomacro\gplbacktext{%
    }%
    \gplgaddtomacro\gplfronttext{%
    }%
    \gplgaddtomacro\gplbacktext{%
    }%
    \gplgaddtomacro\gplfronttext{%
    }%
    \gplgaddtomacro\gplbacktext{%
    }%
    \gplgaddtomacro\gplfronttext{%
    }%
    \gplgaddtomacro\gplbacktext{%
    }%
    \gplgaddtomacro\gplfronttext{%
    }%
    \gplgaddtomacro\gplbacktext{%
    }%
    \gplgaddtomacro\gplfronttext{%
    }%
    \gplgaddtomacro\gplbacktext{%
    }%
    \gplgaddtomacro\gplfronttext{%
    }%
    \gplgaddtomacro\gplbacktext{%
    }%
    \gplgaddtomacro\gplfronttext{%
    }%
    \gplgaddtomacro\gplbacktext{%
    }%
    \gplgaddtomacro\gplfronttext{%
    }%
    \gplbacktext
    \put(0,0){\includegraphics[width={396.00bp},height={170.00bp}]{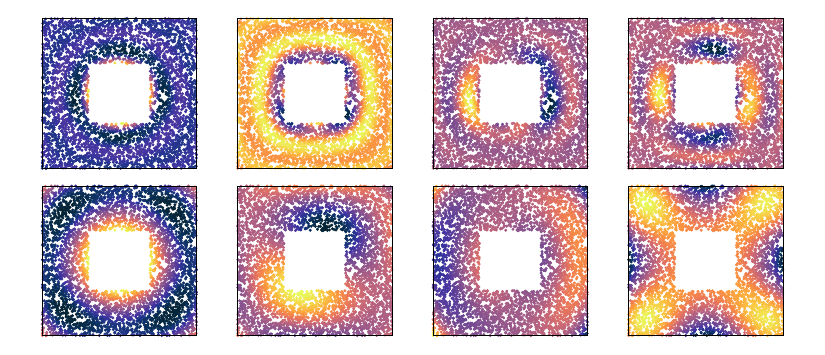}}%
    \gplfronttext
  \end{picture}%
\endgroup

  \caption{Sample columns from $\bm{Q}$ extracted from the column space of
  $\tilde{\bm{\Lambda}} = \bm{\Sigma}^{-1} \tilde{\bm{C}}$ above for a Gaussian
  covariance function $K(\bx - \bx') = \exp((\rho^{-1} \norm[2]{\bx - \bx'})^2)$
  and additive noise with variance $\tau^2=0.1$ at uniform random points on
  $[0,1]^2$ minus a missing region in the center. The top row shows selected
  columns of $\bm{Q}$ with $\rho=0.25$, and the bottom row for $\rho=0.75$.}
  \label{fig:Q_demo}
\end{figure}

We will describe two forms of error analysis. First, we will presume a
proper expansion using a suitable set of interpolation nodes and basis functions
$\set{T_{j}(\bx)}$, presumed to be uniformly bounded by one (as in the case of
Chebyshev polynomials, for example). In this setting, one can establish bounds
depending on an infinity-norm type error control, fundamentally based on the
variations of the error measure $\eps_R$ introduced above. While we do not give
full detailed bounds on the decay of $\eps_R$ with respect to $R$ since it would
naturally be kernel-dependent, the thrust of Theorem \ref{thm:joint_efficiency},
for example, is that if one has an error approximation $\eps_R$ for $\bx'
\mapsto K(\bx - \bx')$ with $\bx \in \mathcal{D}$ and $\bx' \in \Dpred$, then
one can design a $\bm{Q} \in \R^{n \times r}$, $r \leq R$, such that the desired
approximate conditional distributions converge to the true one as fast as
$\eps_R$ decays to zero. If $\Dpred$ is a hyper-rectangle and $K$ is smooth away
from the origin, then, these results can be interpreted as proving exponential
decay in distributional approximation error.

In the second case, we give more generic bounds for an ``improper"
$\tilde{\bm{C}}_R$, where the nodes $\set{\tilde{\bx}_k}$ are not chosen in a
careful or optimal way, but are instead given in an oversampled form with the
intention that a small number of singular vectors from $\tilde{\bm{\Lambda}}_R$
will be selected. Practically speaking, this second regime is the category that
most applications will be in, as even simple data and missing domains
can produce a $\Dpred$ that does not have sufficiently simple structure that,
e.g., Chebyshev, Zernike, or other well-studied orthogonal polynomial families
can be used.  Nonetheless, as both the theoretical results and numerical
demonstrations will show, in settings such as severely oversampled lattice
designs for $\set{\tilde{\bx}_i}_{i=1}^R$, the efficiency cost can be mild. The
important observation in this setting is that with sufficient oversampling and
coverage in the improper interpolation points, one can recover the proper
$\eps_R$-type bounds of Theorem \ref{thm:joint_efficiency} with only a mild cost
of the condition number of an internal interpolation matrix. Considering that a
quadrature rule is in some sense being discovered in this approach, however, we
are reluctant to claim that those more general results prove exponential
convergence like the ideal ones do.

The fundamental error metric that we use to quantify approximation error is the
\emph{Kullback-Leibler divergence}, denoted $\kl{P}{Q}$ for two distributions
$P$ and $Q$, and is defined in general as $\kl{P}{Q} = \int \log \frac{\dif
P(x)}{\dif Q(x)} \dif P(x)$ and interpreted as the expected value of a
log-likelihood ratio test statistic. For two multivariate normal distributions,
this is given by
\begin{equation}\label{eq:kl}
  \kl{N(\bmu, \bS)}{N(\bmu', \bS')}
  =
  \frac{1}{2}\left[
    \mathrm{tr}\!\left(
      (\bS')^{-1}\bS
    \right)
    -
    m
    +
    (\bmu'-\bmu)^T(\bS')^{-1}(\bmu'-\bmu)
    +
    \log\frac{\det\bS'}{\det\bS}
  \right].
\end{equation}
In our particular case, we will compare the multivariate normal distribution
$\bm{f} \sv \by$, where $\bm{f} = [f(\bx_1'), ..., f(\bx_m')]$ is a vector of
un-observed values of the process \emph{without} noise, and the distribution
$\bm{f} \sv \bm{Q}^T \by$.

\subsection{Efficiency analysis with expansion error bounds} \label{sec:eps_analysis}

Due to the presumed smoothness of the covariance function and domain-dependent
effects, when (\ref{eq:C_approx}) can effectively be assembled to high accuracy
(in the sense that an expansion order $R$ can be selected so that $\eps_R$ as
defined above is small for small $R$), the column space of
$\tilde{\bm{\Lambda}}_R$ will nearly contain the column space of $\bm{\Lambda} =
\bm{\Sigma}^{-1} \bm{C}$, where $\bm{C} = K(\mathcal{S}, \Spred)$ is the
cross-covariance between data locations and an arbitrary collection of points in
$\Dpred$. The following result gives an algebraic-type characterization of the
efficiency of conditioning on $\bm{s} = \bm{Q}^T \by$ when predicting at an
arbitrary location $\bx_0 \in \Dpred$.

\begin{theorem} \label{thm:dist}
  Let $\blam(\bx_0)$ be the Kriging weights for predicting $f(\bx_0)$, $\bx_0
  \in \Dpred$, given $\by$, so that $\blam(\bx_0)=\bm{\Sigma}^{-1}\bm c(\bx_0)$
  and $\bm c(\bx_0) = K(\mathcal{S}, \bx_0)$ is the cross covariance. If
  $\tilde{\bm{\Lambda}}_R =\bm Q \bm R$ as above is the machine-precision
  factorization of $\tilde{\bm{\Lambda}}_R =\bm{\Sigma}^{-1}\tilde{\bm C}_R$
  (meaning that degenerate singular vectors have been dropped), then
  \begin{equation*} 
    \text{\emph{dist}}(\blam(\bx_0), \mathrm{sp}(\bm{Q})) \leq \norm[2]{\bm{\Sigma}^{-1}}\,\sqrt n\,\eps_R,
  \end{equation*}
  where $\eps_R$ is given by
  \begin{equation*}
    \eps_R
    = \max_{i=1,\ldots,n}\ \sum_{j > R}\bigl|\beta_{i,j}\bigr|
  \end{equation*}
  and is the largest $L^{\infty}(\Dpred)$ upper bound for approximating each
  $\bx' \mapsto K(\bx_i - \bx')$ on $\Dpred$.
\end{theorem}
\begin{toappendix}
\begin{proof}[Proof of Theorem \ref{thm:dist}]
Define $\bm t(\bx_0) = [T_j(\bx_0)]_{j=1}^R$, and $\bm u(\bx_0) = \bm V^{-T}\bm
t(\bx_0)$, and set $\bm c_1 = \tilde{\bm C}_R\,\bm u(\bx_0)$ and $\bm c_2
= \bm c(\bx_0) - \bm c_1$, so that $\bm c(\bx_0)=\bm c_1+\bm c_2$ and $\bm{c}_1
\in \mathrm{sp}(\tilde{\bm C}_R)$. Next, for each $1 \leq i \leq n$, note by
smoothness that we may expand $K(\bx_i - \bx)$, $\bx \in \Dpred$, as $K(\bx_i -
\bx)=\sum_{j=1}^\infty \beta_{i,j}\,T_{j}(\bx)$, and that 
\begin{equation*}
  |(\bm c_2)_i|
  =\Bigl|\sum_{j > R}\beta_{i,j}\,T_{j}(\bx_0)\Bigr|
  \le \sum_{j > R}\bigl|\beta_{i,j}\bigr|
  \le \eps_R.
\end{equation*}
Therefore we see that $\norm[2]{\bm c_2}\le \sqrt n\,\eps_R$. Now note that
$\blam(\bx_0)=\bm{\Sigma}^{-1}\bm c(\bx_0)=\bm{\Sigma}^{-1}\bm c_1+\bm{\Sigma}^{-1}\bm
c_2$. Since $\tilde{\bm\Lambda}_R =\bm Q\bm R$, $\bm{\Sigma}^{-1}\bm
c_1=\tilde{\bm\Lambda}_R \,\bm u(\bx_0)\in\mathrm{sp}(\bm Q)$. Hence
\begin{equation*}
  \mathrm{dist}\bigl(\blam(\bx_0),\mathrm{sp}(\bm Q)\bigr)
  \le \norm[2]{\bm{\Sigma}^{-1}\bm c_2}
  \le \norm[2]{\bm{\Sigma}^{-1}}\,\norm[2]{\bm c_2}
  \le \norm[2]{\bm{\Sigma}^{-1}}\,\sqrt n\,\eps_R,
\end{equation*}
completing the proof.
\end{proof}
\end{toappendix}
The key intuition of this result is that if the kernel function $K$ is smooth
away from the origin, by the ``dual-form" of the Kriging equation
(\ref{eq:krige-dual}) we see that the function $\bx_0 \mapsto \blam(\bx_0)$ is
smooth so long as the data locations are well-separated from $\bx_0$. In forming
an accurate basis expansion for the kernel function, then, we see that an
analogous rapidly converging expansion representation is given for $\bx_0
\mapsto \blam(\bx_0)$ and Kriging weights can be approximated to extremely high
accuracy uniformly on the domain.  In advance of the following results, we
establish the following notation for conditional means and variances for this
section. Recalling that $\bm{s} = \bm{Q}^T \by$ is our low-dimensional quantity,
define
\begin{align}
  \bm{\mu}_{f \sv y}
  &= E[\bm f \sv \by]
  = \bm C^T \bm{\Sigma}^{-1}\by,
  \label{eq:setup1}
  \\
  \bm{\Sigma}_{f \sv y}
  &= V[\bm f \sv \by]
  = K(\mathcal{S}', \mathcal{S}') - \bm C^T \bm{\Sigma}^{-1}\bm C,
  \label{eq:setup2}
  \\
  \bm{\mu}_{f \sv s}
  &= E[\bm f \sv \bm s]
  = \bm C^T \bm Q(\bm Q^T \bm{\Sigma}\bm Q)^{-1}\bm Q^T\by,
  \label{eq:setup3}
  \\
  \bm{\Sigma}_{f \sv s}
  &= V[\bm f \sv \bm s]
  = K(\mathcal{S}', \mathcal{S}')
  - \bm C^T \bm 
  Q(\bm Q^T \bm{\Sigma}\bm Q)^{-1}\bm Q^T \bm C.
  \label{eq:setup4}
\end{align}
Armed with this result and notation, we now provide a result quantifying the
more direct statistical question of how well conditional distributions are
approximated.

\begin{theorem} \label{thm:joint_efficiency}
  Under the setting of Theorem~\ref{thm:dist}, we have
  \begin{align*}
    \norm[\infty]{E[\bm f \sv \by] - E[\bm f \sv \bm \bs]}
    &\le
    \sqrt{n}
    \eps_R
    \sqrt{\kappa(\bS)
    \norm[2]{\bS^{-1}}}\,
    \sqrt{\by^T\bm{\Sigma}^{-1}\by},
    \\
    \norm[2]{V[\bm f \sv \by] - V[\bm f \sv \bs]}
    &\le
    m n \eps_R^2
    \,
    \kappa(\bS)
    \norm[2]{\bS^{-1}},
  \end{align*}
  where $\kappa(\bm{A}) = \norm[2]{\bm{A}} \norm[2]{\bm{A}^{-1}}$ is the
  matrix condition number. Consequently, we have
  \begin{equation*}
    D_{\mathrm{KL}}\!\bigl(\bm f \sv \by,\ \bm f \sv \bm \bs \bigr)
    \le
    \frac{m\,\kappa(\bm{\Sigma})\,n\,\eps_R^2}{2}\,
    \norm[2]{\bm\Sigma_{f \sv y}^{-1}}\,
    \norm[2]{\bm{\Sigma}^{-1}}\,
    \bigl(1 + \by^T \bm{\Sigma}^{-1} \by\bigr).
  \end{equation*}
\end{theorem}
\begin{toappendix}
\begin{proof}[Proof of Theorem \ref{thm:joint_efficiency}]
For each $k\in\{1,\ldots,m\}$, write $\bm c_k=K(\mathcal S,\bx'_k)$ and
$\blam_k=\bm{\Sigma}^{-1}\bm c_k$.  Define
\begin{equation*}
  \hat{\blam}_k = \bm Q(\bm Q^T\bm{\Sigma}\bm Q)^{-1}\bm Q^T \bm c_k \in \mathrm{sp}(\bm Q),
  \qquad
  \bm r_k = \blam_k-\hat{\blam}_k.
\end{equation*}
By Equations (\ref{eq:setup1} - \ref{eq:setup4}), we have $(E[\bm f \sv \by]-E[\bm f \sv
\bm s])_k=\bm r_k^T\by$, and we observe that
\begin{equation*}
  |\bm r_k^T\by|
  = |\bm r_k^T \bm{\Sigma}^{1/2}\bm{\Sigma}^{-1/2}\by|
  \le \sqrt{\bm r_k^T\bm{\Sigma}\,\bm r_k}\ \sqrt{\by^T\bm{\Sigma}^{-1}\by}.
\end{equation*}
Since $\hat{\blam}_k\in\mathrm{sp}(\bm Q)$ is the $\bm{\Sigma}$-orthogonal
projection of $\blam_k$ onto $\mathrm{sp}(\bm Q)$,
\begin{equation*}
  \bm r_k^T\bm{\Sigma}\,\bm r_k
  = \min_{\bm w\in\mathrm{sp}(\bm Q)}(\bm w-\blam_k)^T\bm{\Sigma}(\bm w-\blam_k)
  \le \norm[2]{\bm{\Sigma}}\,\mathrm{dist}(\blam_k,\mathrm{sp}(\bm Q))^2.
\end{equation*}
Applying Theorem~\ref{thm:dist} gives $\mathrm{dist}(\blam_k,\mathrm{sp}(\bm Q))\le \norm[2]{\bm{\Sigma}^{-1}}\sqrt n\,\eps_R$, and therefore
\begin{equation*}
  |\bm r_k^T\by|
  \le
  \sqrt{\norm[2]{\bm{\Sigma}}}\,\norm[2]{\bm{\Sigma}^{-1}}\,\sqrt n\,\eps_R\,
  \sqrt{\by^T\bm{\Sigma}^{-1}\by}.
\end{equation*}
Taking the maximum over $k$ yields the stated $\ell_\infty$ bound on the conditional means.

For the conditional covariances, define the $\bm{\Sigma}$-orthogonal projector onto $\mathrm{sp}(\bm Q)$
\begin{equation*}
  \bm P = \bm Q(\bm Q^T\bm{\Sigma}\bm Q)^{-1}\bm Q^T\bm{\Sigma},
  \qquad
  \bm R_f = (\bm I-\bm P)\bm\Lambda = [\bm r_1\ \cdots\ \bm r_m].
\end{equation*}
Using $\bm C=\bm{\Sigma}\bm\Lambda$ and $\bm P^T\bm{\Sigma}=\bm{\Sigma}\bm P$,
\begin{align*}
  V[\bm f \sv \bm s]-V[\bm f \sv \by]
  &=
  \bm C^T\bm{\Sigma}^{-1}\bm C - \bm C^T \bm Q(\bm Q^T\bm{\Sigma}\bm Q)^{-1}\bm Q^T \bm C \\
  &=
  \bm\Lambda^T\bm{\Sigma}\,\bm\Lambda - \bm\Lambda^T\bm{\Sigma}\,\bm P\,\bm\Lambda
  =
  \bm\Lambda^T\bm{\Sigma}(\bm I-\bm P)\bm\Lambda
  =
  \bm R_f^T\bm{\Sigma}\,\bm R_f \succeq 0.
\end{align*}
Hence
\begin{equation*}
  \norm[2]{V[\bm f \sv \bm s]-V[\bm f \sv \by]}
  \le \mathrm{tr}\!\bigl(\bm R_f^T\bm{\Sigma}\,\bm R_f\bigr)
  = \sum_{k=1}^m\bm r_k^T\bm{\Sigma}\,\bm r_k
  \le m\,\norm[2]{\bm{\Sigma}}\,\norm[2]{\bm{\Sigma}^{-1}}^2\,n\,\eps_R^2,
\end{equation*}
which gives the stated bound on conditional variance approximation.

For the KL divergence, note $\bm\Sigma_{f \sv s}=\bm\Sigma_{f \sv
y}+(\bm\Sigma_{f \sv s}-\bm\Sigma_{f \sv y})$ with $\bm\Sigma_{f \sv
s}-\bm\Sigma_{f \sv y}\succeq 0$. Then by observing that $\bm\Sigma_{f \sv
s}^{-1}\preceq \bm\Sigma_{f \sv y}^{-1}$, and noting that $\log\det(\bm I+\bm
A)\le \mathrm{tr}(\bm A)$ for $\bm A\succeq 0$, we have that
\begin{align*}
  D_{\mathrm{KL}}\!\bigl(\bm f \sv \by,\ \bm f \sv \bm s\bigr)
  &=
  \tfrac12\Bigl(
    \mathrm{tr}(\bm\Sigma_{f \sv s}^{-1}\bm\Sigma_{f \sv y}) - m
    + (\bm\mu_{f \sv y}-\bm\mu_{f \sv s})^T\bm\Sigma_{f \sv s}^{-1}(\bm\mu_{f \sv y}-\bm\mu_{f \sv s})
    + \log\frac{\det\bm\Sigma_{f \sv s}}{\det\bm\Sigma_{f \sv y}}
  \Bigr) \\
  &\le
  \tfrac12\Bigl(
    \norm[2]{\bm\Sigma_{f \sv y}^{-1}}\,\norm[2]{\bm\mu_{f \sv y}-\bm\mu_{f \sv s}}^2
    + \norm[2]{\bm\Sigma_{f \sv y}^{-1}}\,\mathrm{tr}(\bm\Sigma_{f \sv s}-\bm\Sigma_{f \sv y})
  \Bigr).
\end{align*}
Using $\norm[2]{\bm\mu_{f \sv y}-\bm\mu_{f \sv s}}^2\le m\norm[\infty]{\bm\mu_{f \sv y}-\bm\mu_{f \sv s}}^2$ and the two bounds proved above yields the stated KL bound.
\end{proof}
\end{toappendix}
Several comments about the above results are in order. First, we note that if
$\eps_R$ is small, then the column space of $\tilde{\bm{C}}_R$ provides a highly
accurate approximation to that of $\bm{C}$, and so if $\tilde{\bm{\Lambda}}_R$
is approximated to machine precision and $\norm[2]{\bm{\Sigma}^{-1}}$ is bounded
then these efficiency bounds demonstrate conditioning on $\bm{s} = \bm{Q}^T \by$
to be effectively machine-precision equivalent to conditioning on $\by$. What is
particularly notable about this is that these bounds are \emph{sharper} as, for
example, measurement noise variance $\tau^2$ increases, since
$\norm[2]{\bm{\Sigma}^{-1}}$ is controlled more tightly. This feature can be seen in
Figure \ref{fig:intro_knn_compare}: while measurement noise worsens the
efficiency of nearest neighbor-based prediction, it actually \emph{improves} the
accuracy of this formulation by improving on the condition number
$\kappa(\bm{\Sigma})$.

To close this section, we provide a closely related result related to an
incomplete factorization $\tilde{\bm{\Lambda}}_R$ as
\begin{equation} \label{eq:Lam_generic}
  \tilde{\bm{\Lambda}} 
  = 
  \bm{\Sigma}^{-1} \tilde{\bm{C}} 
  = 
  \mat{\bm{Q} & \bm{Q}_{\perp}} \bm{R}, 
\end{equation}
where $\bm{Q} \in \R^{n \times r}$ is presumed to be a basis for the dominant
$r$ singular vectors of $\bm{\Lambda}$. Such a $\bm{Q}$ can be rapidly extracted
in settings where linear systems pertaining to $\bm{\Sigma}$ can be solved
rapidly using randomized linear algebra routines, which will be discussed in
detail in the next section.

\begin{theorem} \label{thm:generic_efficiency}
  Let data $\by$ be measured at locations $\mathcal{S}$ as above, and let
  $\tilde{\bm{\Lambda}}$ be factored as in (\ref{eq:Lam_generic}), so that
  $\bm{Q}$ is simply a basis for its first $r$ singular vectors. Then
  \begin{align*} 
    \norm[2]{E[\bm{f} \sv \by] - E[\bm{f} \sv \bs]}
    &\le
    \sqrt{m}\,\sqrt{\norm[2]{\bm{\Sigma}}}\,\sigma_{r+1}(\tilde{\bm{\Lambda}})\,
    \sqrt{\by^T\bm{\Sigma}^{-1}\by}
    \\
    \norm[2]{V[\bm{f} \sv \by] - V[\bm{f} \sv \bs]}
    &\le
    m\,\norm[2]{\bm{\Sigma}}\,\sigma_{r+1}(\tilde{\bm{\Lambda}})^2.
  \end{align*}
  Consequently, we have
  \begin{equation*}
    D_{\mathrm{KL}}\!\bigl(\bm f \sv \by,\ \bm f \sv \bm \bs\bigr)
    \le
    \frac{m}{2}\,\norm[2]{\bm\Sigma_{f \sv y}^{-1}}\,
    \norm[2]{\bm{\Sigma}}\,\sigma_{r+1}(\tilde{\bm{\Lambda}})^2
    \bigl(1+\by^T\bm{\Sigma}^{-1}\by\bigr).
  \end{equation*}
\end{theorem}
\begin{toappendix}
\begin{proof}[Proof of Theorem \ref{thm:generic_efficiency}]
Define the $\bm\Sigma$-orthogonal projector onto $\mathrm{sp}(\bm Q)$ by
$
  \bm P = \bm Q(\bm Q^T\bm\Sigma\bm Q)^{-1}\bm Q^T\bm\Sigma.
$
Writing $\tilde{\bm\Lambda}=[\tilde{\blam}_1\ \cdots\ \tilde{\blam}_m]$ and
$\bm R_f=(\bm I-\bm P)\tilde{\bm\Lambda}=[\bm r_1\ \cdots\ \bm r_m]$, the same
algebra as in Theorem~\ref{thm:joint_efficiency} gives
\begin{equation*}
  E[\bm f \sv \by]-E[\bm f \sv \bm s]=\bm R_f^T\by,
  \qquad
  V[\bm f \sv \bm s]-V[\bm f \sv \by]=\bm R_f^T\bm\Sigma\,\bm R_f \succeq 0.
\end{equation*}
For the mean bound, for each $k$,
\begin{equation*}
  |\bm r_k^T\by|
  = |\bm r_k^T \bm\Sigma^{1/2}\bm\Sigma^{-1/2}\by|
  \le \sqrt{\bm r_k^T\bm\Sigma\,\bm r_k}\ \sqrt{\by^T\bm\Sigma^{-1}\by}.
\end{equation*}
Using the identity $(\bm I-\bm P)\tilde{\blam}_k = (\bm I-\bm P)(\bm I-\bm Q\bm
Q^T)\tilde{\blam}_k$, we have
\begin{equation*}
  \bm r_k^T\bm\Sigma\,\bm r_k 
  \le \norm[2]{\bm\Sigma} \norm[2]{(\bm I-\bm Q\bm Q^T)\tilde{\blam}_k}^2
  \le \norm[2]{\bm\Sigma} \sigma_{r+1}(\tilde{\bm{\Lambda}})^2,
\end{equation*}
and therefore
\begin{equation*}
  |\bm r_k^T\by|
  \le \sqrt{\norm[2]{\bm\Sigma}}\,\sigma_{r+1}(\tilde{\bm{\Lambda}})\,\sqrt{\by^T\bm\Sigma^{-1}\by}.
\end{equation*} 
Summing over $k$ yields
\begin{equation*}
  \norm[2]{E[\bm f \sv \by]-E[\bm f \sv \bm s]}^2
  = \sum_{k=1}^m (\bm r_k^T\by)^2
  \le m\,\norm[2]{\bm\Sigma}\,\sigma_{r+1}(\tilde{\bm\Lambda})^2\,\by^T\bm\Sigma^{-1}\by,
\end{equation*}
which proves the conditional mean bound.

For the variance bound, using $V[\bm f \sv \bm s]-V[\bm f \sv \by]=\bm R_f^T\bm\Sigma\,\bm R_f\succeq 0$,
\begin{equation*}
  \norm[2]{V[\bm f \sv \by]-V[\bm f \sv \bm s]}
  = \norm[2]{\bm R_f^T\bm\Sigma\,\bm R_f}
  \le \mathrm{tr}(\bm R_f^T\bm\Sigma\,\bm R_f)
  = \sum_{k=1}^m \bm r_k^T\bm\Sigma\,\bm r_k
  \le m\,\norm[2]{\bm\Sigma}\,\sigma_{r+1}(\tilde{\bm\Lambda})^2,
\end{equation*}
which proves the conditional covariance bound.

Finally, the KL bound follows by the same final step as in Theorem~\ref{thm:joint_efficiency}:
since $\bm\Sigma_{f \sv s}-\bm\Sigma_{f \sv y}\succeq 0$, we have $\bm\Sigma_{f \sv s}^{-1}\preceq \bm\Sigma_{f \sv y}^{-1}$ and
\begin{equation*}
  D_{\mathrm{KL}}\!\bigl(\bm f \sv \by,\ \bm f \sv \bm s\bigr)
  \le \tfrac12\Bigl(
    \norm[2]{\bm\Sigma_{f \sv y}^{-1}}\,\norm[2]{\bm\mu_{f \sv y}-\bm\mu_{f \sv s}}^2
    + \norm[2]{\bm\Sigma_{f \sv y}^{-1}}\,\mathrm{tr}(\bm\Sigma_{f \sv s}-\bm\Sigma_{f \sv y})
  \Bigr).
\end{equation*}
Using $\norm[2]{\bm\mu_{f \sv y}-\bm\mu_{f \sv s}}^2 \le m\,\norm[\infty]{\bm\mu_{f \sv y}-\bm\mu_{f \sv s}}^2$,
the mean bound above, and $\mathrm{tr}(\bm\Sigma_{f\mid s}-\bm\Sigma_{f\mid y})
= \mathrm{tr}(\bm R_f^T\bm\Sigma\,\bm R_f)\le m\,\norm[2]{\bm\Sigma}\,\sigma_{r+1}(\tilde{\bm\Lambda})^2$
yields the displayed KL inequality.
\end{proof}
\end{toappendix}

\subsection{Efficiency analysis with arbitrary interpolation points}
\label{sec:generic_error_theory}

In the above analysis, it is presumed that the interpolation nodes
$\set{\tilde{\bx}_j}_{j=1}^R$ are selected to be, e.g., Chebyshev nodes. In
practice, however, it may be the case that $\Dpred$ has an unusual structure,
and so optimal or near-optimal interpolating basis functions $\set{T_{j}}$ and
nodes $\set{\tilde{\bx}_j}_{j=1}^R$ may not be convenient or feasible to
compute. A practical approach that we expect will be common to remedy this issue
is to heavily \emph{oversample} on $\Dpred$ in a sub-optimal way, for example
with a very dense regular grid, and using partial factorization methods as were
the setting of Theorem \ref{thm:generic_efficiency}. Define
$\set{\bx^g_l}_{l=1}^{L}$ to be such a (presumably oversampled) collection of
points. In the following result, we show that the efficiency bounds of Theorem
\ref{thm:joint_efficiency} can still be achieved, but with the balancing
multiplicative penalty of $\kappa^g_R$, the condition number of the
interpolation matrix from $\set{\bx^g_l}_{l=1}^{L}$ to
$\set{\tilde{\bx}_j}_{j=1}^R$.
\begin{theorem} \label{thm:kappa_efficiency}
Let $\set{\bx^g_l}_{l=1}^{L} \subset \Dpred$ be an arbitrary collection of
points such that the interpolation matrix $\bm{U}^g$, defined by
\begin{equation*} 
  \bm{U}^g = \bm E^g \bm V^{-1},
  \qquad
  \bm{V}_{k,\ell}=T_\ell(\tilde{\bx}_k),
  \qquad
  \bm{E}^g_{l,\ell}=T_\ell(\bx^g_l),
\end{equation*}
has full column rank. Next, define the analogous partial decomposition 
\begin{equation} \label{eq:general_lowrank_lam}
  \tilde{\bm{\Lambda}}_g = \bm{\Sigma}^{-1} \tilde{\bm{C}}_g = 
  \mat{\bm{Q} & \bm{Q}_{\perp}} \bm{R},
\end{equation}
where $\tilde{\bm{C}}_g = [K(\bx_j - \bx^g_l)]_{j,l=1}^{n,L}$ is the
cross-covariance between given data locations $\set{\bx_j}_{j=1}^n$ and the
general interpolation points $\set{\bx^g_l}$ and $\bm{Q} \in \R^{n \times r}$ is
presumed to be a basis for the first $r$ singular vectors of
$\tilde{\bm{\Lambda}}_g$. Then letting $\bm{s} = \bm{Q}^T \by$,
\begin{equation*} 
\begin{aligned}
  \norm[2]{E[\bm f \sv \by] - E[\bm f \sv \bm{s}]}
  &\le
  \kappa_R^g
  \sqrt{m}\,\sqrt{\norm[2]{\bm{\Sigma}}}\,\sigma_{r+1}(\tilde{\bm{\Lambda}}_R)\,
  \sqrt{\by^T\bm{\Sigma}^{-1}\by},
  \\
  \norm[2]{V[\bm f \sv \by] - V[\bm f \sv \bm{s}]}
  &\le
  (\kappa_R^g)^2
  m\,\norm[2]{\bm{\Sigma}}\,\sigma_{r+1}(\tilde{\bm{\Lambda}}_R)^2,
\end{aligned}
\end{equation*}
where $\kappa^g_R = \sigma_1(\bm{U}^g)/\sigma_R(\bm{U}^g)$ is the condition
number for recovering the proper interpolation points
$\set{\tilde{\bx}_k}_{k=1}^{R}$ from the given nodes $\set{\bx^g_l}_{l=1}^{L}$.
Using the same shorthand notation as above, the KL-divergence between $\bm{f}
\sv \by$ and $\bm{f} \sv \bm{s}$ is bounded as
\begin{equation*} 
  D_{\mathrm{KL}}\!\bigl(\bm f \sv \by,\ \bm f \sv \bm{s} \bigr)
  \le
  \frac{m}{2}\,\norm[2]{\bm\Sigma_{f \sv y}^{-1}}\,
  \norm[2]{\bm{\Sigma}}\,(\kappa_R^g)^2\,\sigma_{r+1}(\tilde{\bm{\Lambda}}_R)^2
  \bigl(1+\by^T\bm{\Sigma}^{-1}\by\bigr).
\end{equation*}
\end{theorem}
\begin{toappendix}
\begin{proof}[Proof of Theorem \ref{thm:kappa_efficiency}]
By the same interpolation construction as in
Section~\ref{sec:low-rank-crosscov}, $\tilde{\bm C}_g = \tilde{\bm C}_R\,(\bm
U^g)^T$, and so
$
  \tilde{\bm\Lambda}_g
  = \bm\Sigma^{-1}\tilde{\bm C}_g
  = \tilde{\bm\Lambda}_R\,(\bm U^g)^T.
$
Additionally, by the assumption that $\bm{U}^g$ is full rank, we can multiply by
the pseudo-inverse $((\bm U^g)^T)^{\dagger}$ to obtain
$
  \tilde{\bm\Lambda}_R = \tilde{\bm\Lambda}_g\,((\bm U^g)^T)^{\dagger}.
$
Now, let $\bm P = \bm Q (\bm Q^T\bS \bm Q)^{-1}\bm Q^T \bS$ be the
$\bm\Sigma$-orthogonal projector onto $\mathrm{sp}(\bm Q)$, and let $\bm r_k =
(\bm I-\bm P)\tilde{\blam}^R_k$ denote the columns of the residual matrix $\bm
R_f$, where $\tilde{\blam}_k^R$ is the $k$-th column of $\tilde{\bm{\Lambda}}_R$. As in
Theorem~\ref{thm:generic_efficiency}, we see that
\begin{equation*} 
\bm r_k^T\bm\Sigma\,\bm r_k \le \norm[2]{\bm\Sigma}
\norm[2]{(\bm I-\bm Q\bm Q^T)\tilde{\blam}^R_k}^2 \le \norm[2]{\bm\Sigma}
\norm[2]{(\bm I-\bm Q\bm Q^T)\tilde{\bm\Lambda}_R}^2.
\end{equation*}
This Euclidean truncation error is bounded as
\begin{align*}
  \norm[2]{(\bm I-\bm Q\bm Q^T)\tilde{\bm\Lambda}_R} 
  \leq 
  \norm[2]{(\bm I-\bm Q\bm Q^T)\tilde{\bm\Lambda}_g} \, \norm[2]{((\bm U^g)^T)^{\dagger}} 
  \leq 
  \sigma_{r+1}(\tilde{\bm\Lambda}_g)\,\frac{1}{\sigma_R(\bm U^g)}
  \leq
  \kappa_R^g \sigma_{r+1}(\tilde{\bm{\Lambda}}_R).
\end{align*}
Following the remaining steps of Theorem~\ref{thm:generic_efficiency} with this
looser bound on $\bm r_k^T\bm\Sigma\,\bm r_k$ yields the results.
\end{proof}
\end{toappendix}

Consider, as an intuitive example, the case where $\set{\bx^g_j}$ are a very
dense regular grid on $\Dpred$. This result indicates that when the fineness of
this grid is such that there are grid points very near to the proper
interpolation points $\set{\tilde{\bx}_j}_{j=1}^R$, the condition number will be
small as the column space of that $\tilde{\bm{C}}_R$ will be well-represented
with a minimal sharpness cost. As $R$ is increased, however, it may be that a
dense regular grid of points (or something less regular that has insufficient
coverage on parts of the domain) will lead to an enormous $\kappa^g_R$, and thus
the bound above will not be informative.

\section{Accelerated machine-precision conditioning} \label{sec:computation}

As the last section explains and the next section will numerically demonstrate,
particularly for structured prediction problems in which the prediction domain
is a large well-behaved region there are very low-dimensional random variables
that contain effectively \emph{all} of the relevant conditioning information of
a given data vector $\by$. And importantly, since this variety of dimension
reduction for prediction depends only on smoothness properties of the kernel
away from the origin, it is not impacted by measurement noise and other features
that can make $k$-nearest neighbor prediction much less effective. In this
section, we will now discuss and survey computational procedures to accelerate
the extraction of these low-dimensional conditioning variables, as well as the
necessary pre-computations for fast online evaluation of conditional
distributions in $\bO(1)$ time.

In section \ref{sec:fastQ}, we provide a discussion of the use of either
partially pivoted factorizations or randomized algorithms to rapidly extract the
matrix $\bm{Q}$ of dominant singular vectors of $\tilde{\bm{\Lambda}}$ in
settings where linear solves $\bv \mapsto \bS^{-1} \bv$ can be computed quickly.
Then, in Section \ref{sec:online}, we will discuss the necessary subsequent
precomputations for $\bO(1)$-cost online conditional distributions.

\subsection{Rapid assembly of $\bm{Q}$} \label{sec:fastQ}

As established in the previous section, if one can select reasonable
interpolation nodes $\set{\tilde{\bx}_k}_{k=1}^m$ and extract the dominant
singular vector space of $\tilde{\bm{\Lambda}} = \bS^{-1} \tilde{\bm{C}}$ so
that $\norm[2]{\tilde{\bm{\Lambda}} - \bm{Q} \bm{R}} \approx
\sigma_{k+1}(\tilde{\bm{\Lambda}})$, then one can achieve very accurate
approximations to the conditional distribution $f(\bx) \sv \by$ with $f(\bx) \sv
\bm{Q}^T \by$ for $\bx \in \Dpred$. The purpose of this section is to describe
two procedures for doing so and outline the associated challenges and potential
design choices.

Let us first assume that $\tilde{\bm{C}}$ is sufficiently rank-deficient that it
can be factorized to machine precision as $\tilde{\bm{C}} \approx \bm{U}_c
\bm{V}_c^T$. In the setting where its rank is manageable, there are several
appealing options for rapidly obtaining such a factorization. If the covariance
function $K$ is easy to evaluate directly such that individual indexing of
$\tilde{\bm{C}}$ is fast, a greedy pivoting factorization such as the
\emph{adaptive cross approximation} \cite{bebendorf2000} can be used to obtain
the factorization in $\bO(\max(n, m) r^2)$, where $r$ is the (numerical) rank of
$\tilde{\bm{C}}$. Alternatively, if matrix-vector products with $\tilde{\bm{C}}$
are more convenient than individual indexing, a randomized factorization
approach can be used instead \cite{halko2011,tropp2023}. In particular, letting
$\bm{\Omega} \in \R^{n \times (r+p)}$ with $\bm{\Omega}_{j,k} \sim N(0,1)$ be
the \emph{sketching matrix} with oversampling parameter $p\in\mathbb{N}$, one
can obtain a near-optimal rank-$r$ approximation to $\tilde{\bm{C}}$ from within
the dominant subspaces of $\tilde{\bm{C}} \bm{\Omega}$ with very favorable
control of sub-optimality due to randomization \cite{halko2011,tropp2023}. Since
that matrix has $r + p$ columns, that method will cost $\bO(T r^2)$ work plus a
potential small additional cost for converting from one low-rank representation
to another (e.g., converting an unstructured $\bm{U}_c \bm{V}_c^T$ to a partial
SVD) \cite{halko2011}, where $T$ is the cost of the matrix-vector product $\bv
\mapsto \tilde{\bm{C}} \bv$.

Continuing in the setting where the low-rank approximation $\tilde{\bm{C}}
\approx \bm{U}_c \bm{V}_c^T$ can be directly formed to machine precision, the
next step involves re-factorizing the matrix $\bS^{-1} \bm{U}_c \bm{V}_c^T =
\bm{Q} \bm{R}$ to obtain an orthonormal basis for its non-degenerate column
space. The dominant cost of this task will be the linear solves $\bv \mapsto
\bS^{-1} \bv$. In the case where the covariance function $K$ is globally smooth,
such as when $K$ is the Gaussian or squared exponential covariance function,
$K(\mathcal{S}, \mathcal{S})$ will have low global rank, and so one of the
methods described in the last section can be applied just the same to obtain an
approximation like $\bS = K(\mathcal{S}, \mathcal{S}) + \tau^2 \I \approx \bm{G}
\bm{G}^T + \tau^2 \I$. At that point, $\bS^{-1} \bm{U}_c$ can be computed
rapidly using the Sherman-Morrison-Woodbury formula, and a standard QR
factorization will provide $\bm{Q}$.

The setting where $K$ is \emph{not} globally smooth requires more work and the
specific strategy will depend on the covariance function $K$. Some algorithms
for accelerated operations with $\bS$ offer fast direct solvers: hierarchical
matrices \cite{hackbusch2015}, for example, have been applied in many Gaussian
process settings very successfully
\cite{ambikasaran2015,litvinenko2019,geoga2020}. Closely related
skeletonization-based solvers \cite{minden2017,baugh2018} have shown similar
results. For a much more thoughtful and complete survey of fast direct solvers,
we refer the reader to \cite{martinsson2025}. Much more frequently, however,
fast algorithms only accelerate the matrix-vector product $\bv \mapsto \bS \bv$,
requiring the use of an iterative solver such as the conjugate gradient method.
Thankfully, sparse approximations to $\bS^{-1}$ based on, e.g., the Vecchia
approximation \cite{vecchia1988} provide very effective preconditioners
\cite{guinness2019} (see \cite{geoga_vecchia_jl} for a performant library and
example codes), and so if the number of columns $r$ of $\bm{U}_c$ is not
prohibitively high, one may obtain $\bS^{-1} \bm{U}_c$ in this setting in a
tolerable number of matrix-vector products with $\bS$. Algorithms in this
category that provide extremely fast matrix-vector products include the FMM
\cite{greengard1987}, fast Gauss transform \cite{greengard1991}, Ewald summation
\cite{toukmaji1996}, the equispaced Fourier GP algorithm of
\cite{greengard2022}, and countless others.

In the setting where fully factorizing $\tilde{\bm{C}}$ is \emph{not} feasible,
one must directly extract the $r$ dominant singular vectors of
$\tilde{\bm{\Lambda}}$. Randomized sketching-based algorithms that again boil
down to extracting information from $\tilde{\bm{\Lambda}} \bm{\Omega}$ for a
suitable sketching matrix $\bm{\Omega} \in \R^{n \times (r+p)}$ become the
clear best option. The choice of how to compute $\bv \mapsto \bS^{-1} \bv$
can be made again depending on properties of $K$ near the origin and practical
rank considerations. But the process is effectively the same: after designing
some procedure for rapidly evaluating $\bv \mapsto \bS^{-1} \tilde{\bm{C}} \bv$,
a sketching-based approach can be used to obtain the partial singular vector
basis $\bm{Q}$ at the cost of application to a mildly oversampled sketching
matrix $\bm{\Omega} \in \R^{n \times (r+p)}$.

\subsection{Rapid online prediction} \label{sec:online}

The constructions above yield predictive distributions conditioned on the
compressed summary $\mathbf{s}=\mathbf{Q}^\top \mathbf{y}\in\mathbb{R}^r$. To
clarify what must be computed online for a new query location, consider a single
prediction $f_0=f(\bx_0)$ and the pair $(\mathbf{s},f_0)$.  Under the GP model,
we have the joint Gaussian distribution
\begin{equation*}
\begin{bmatrix}
\mathbf{s} \\
f_0
\end{bmatrix}
=
\begin{bmatrix}
\mathbf{Q}^\top \mathbf{y} \\
f(\bx_0)
\end{bmatrix}
\sim
\mathcal{N}\!\left(
\mathbf{0},\;
\begin{bmatrix}
\mathbf{Q}^\top \mathbf{\Sigma} \mathbf{Q} & \mathbf{Q}^\top \bm{c}_0 \\
\bm{c}_0^\top \mathbf{Q} & K(\bx_0,\bx_0)
\end{bmatrix}
\right),
\qquad
\bm{c}_0 = K(\mathcal{S},\bx_0)\in\mathbb{R}^n.
\end{equation*}
Consequently, the $\mathbf{Q}$-conditional predictive distribution at $\bx_0$ is
\begin{equation} \label{eq:s_cond}
f_0 \mid \mathbf{s}
\sim
\mathcal{N}\!\Bigl(
\bm{c}_0^\top \mathbf{Q}\,(\mathbf{Q}^\top \mathbf{\Sigma} \mathbf{Q})^{-1}\mathbf{s},\;
K(\bx_0,\bx_0)-\bm{c}_0^\top \mathbf{Q}\,(\mathbf{Q}^\top \mathbf{\Sigma} \mathbf{Q})^{-1}\mathbf{Q}^\top \bm{c}_0
\Bigr).
\end{equation}
Since $\bm{Q}^T \bS \bm{Q}$ can obviously be precomputed, the online cost of
evaluating the predictive mean and variance at a new query location $x_0$ is
dominated by computing the cross-term $\mathbf{Q}^\top
\bm{c}_0\in\mathbb{R}^r$. If this cross-term can be obtained in
$\mathcal{O}(1)$ work per query location, then the entire evaluation of the
$\mathbf{Q}$-conditional distribution is $\mathcal{O}(1)$ per query.
Similarly to above, the solution we propose to this problem---which depends
crucially on the kernel being smooth away from the origin---is kernel
anterpolation. Define $\psi_k(\bx_0)$ to be the $k$-th entry of $\bm{Q}^T
\bm{c}_0$, so that
\begin{equation} \label{eq:cross_psi}
  \bm{Q}^T \bm{c}_0 = \mat{
    \sum_{j=1}^n Q_{j,1} K(\bx_j - \bx_0)
    \\
    \vdots
    \\
    \sum_{j=1}^n Q_{j,r} K(\bx_j - \bx_0)
  }
  =
  \mat{
  \psi_1(\bx_0)
  \\
  \vdots
  \\
  \psi_r(\bx_0)
  }.
\end{equation}
Since $K$ is smooth away from the origin and $\bx_0$ is presumed well-separated
from $\bx_j \in \mathcal{D}_{\rm{observe}}$, we note that the exact same
expansion of
\begin{equation*} 
  \psi_k(\bx) \approx \sum_{\bm{j} \in [p']^d} \beta^{(k)}_{\bm{j}} T_{\bm{j}}(\bx)
\end{equation*}
can be computed for each index $k = 1, ..., r$. We may therefore precompute
\emph{proxy points} and \emph{proxy charges} (in the terminology of
\cite{jiang2025}), so that
\begin{equation} \label{eq:proxy}
\psi_k(\bx)
= \sum_{i=1}^n Q_{i,k}\,K(\bx_i-\bx)
\;\approx\;
\sum_{j=1}^R \tilde{\beta}^{(k)}_j\,K(\tilde{\bx}_j-\bx),
\qquad
\tilde{\bm{\beta}}^{(k)} = \bm U^\top \bm Q_{:,k},
\end{equation}
where $\{\tilde{\bx}_j\}_{j=1}^R$ are the Chebyshev proxy points and $\bm
U^\top$ is the anterpolation matrix.  And since each of these proxy point and
proxy charge computations can be done offline in $\bO(n)$ work, the total
precomputation cost is $\bO(T r^2)$, where again $T$ is the dominant cost
between a solve with $\bS^{-1}$ and a matrix-vector product with $\bm{C}$, and
the \emph{online} work to obtain (\ref{eq:s_cond}) is $\bO(1)$ using the
anterpolation compression method for the cross-covariances (\ref{eq:cross_psi}).
The entire procedure for all precomputation and online evaluation is given in
Algorithms $3.1$ and $3.2$ respectively.

\begin{algorithm}[ht!] \label{alg:predict_precompute}
\caption{Precomputation of $\bm{Q}$ and proxy expansions for fast prediction}
 \begin{enumerate}
 \item Prepare, using one of the approaches referred to in Section
 \ref{sec:fastQ}, a method for evaluating $\bv \mapsto \bS^{-1} \tilde{\bm{C}}
 \bv$.
 \item Obtain $\tilde{\bm{\Lambda}} = \bS^{-1} \tilde{\bm{C}} \approx \bm{Q}
 \bm{R}$ to the desired accuracy or column count, either by pre-factorizing
 $\tilde{\bm{C}} \approx \bm{U}_c \bm{V}_c^T$ and re-factorizing $\bS^{-1}
 \bm{U}_c$ and $\bm{V}_c$ or directly using a sketching or pivoting
 factorization on $\bS^{-1} \tilde{\bm{C}}$.
 \item Precompute $\bs = \bm{Q}^T \by$ and $\V \bm{s} = \bm{Q}^T \bS \bm{Q}$.
 \item Precompute $\tilde{\bs} = (\bm{Q}^T \bS \bm{Q})^{-1} \bs$.
 \item Precompute the proxy points and charges to compute $\psi_k(\bx)$ of
 $\bm{Q}^T K(S, \bx)$ as in (\ref{eq:proxy}).
 \end{enumerate}
\end{algorithm}

\begin{algorithm}[ht!] \label{alg:predict_online}
\caption{Online prediction at a single location $\bx_0$}
 \begin{enumerate}
 \item Evaluate $\psi_k(\bx_0)$ for $k=1, ..., r$ using (\ref{eq:proxy}).
 \item Using those evaluations, compute the conditional mean with $\sum_{k=1}^r
 \psi_k(\bx_0)\,\tilde{s}_k$.
 \item Again using those evaluations, form cross-covariances $\bm{Q}^T \bm{c}_0$ with
 (\ref{eq:cross_psi}).
 \item Compute the conditional variance using the precomputed quadratic forms,
 and form the $(r+1)$-dimensional joint distribution with (\ref{eq:s_cond}).
 \end{enumerate}
\end{algorithm}

\section{Numerical demonstrations} \label{sec:demo}

In this section, we will now provide a variety of demonstrations of the
properties of this reduced conditioning approach. After demonstrating that the
reduced basis $\bm{Q}$ can be obtained rapidly for even very large problems, we
will give two very different examples of this approach outperforming nearest
neighbor-based prediction. Additionally, we will introduce a method for
approximating large covariance matrices that uses this approach to directly
build an inverse Cholesky factor, similarly to how Vecchia approximations are
used in contemporary practice. Finally, we provide a practical demonstration of
the number of predictions required to sufficiently amortize the cost of
assembling $\bm{Q}$ in comparison with $k$ nearest neighbor-based prediction
\footnote{All code is available at
\texttt{https://github.com/cgeoga/LinearComboConditioning.jl}.}.

\subsection{Runtime cost} \label{sec:runtime}

We begin by verifying the runtime cost of computing the matrix $\bm{Q}$
described above where $K$ is the squared exponential kernel, given
by
\begin{equation} \label{eq:gausskernel}
  K(\bx - \bx') = \exp \left(
    - \rho^{-2} \norm[2]{\bx - \bx'}^2
  \right).
\end{equation}
In particular, we test the runtime cost of a pragmatic procedure that we expect
to be particularly common in real applications: a missing data region with a
simple geometry, but not to the degree that one can precisely compute a family
of orthogonal functions like Chebyshev polynomials to expand on. In particular,
we simulate $n = 2^k$ independent uniform points on $[0,1]^2$ and remove
locations $\bx$ such that $\norm{\bx - [1/2, 1/2]} < \delta$ for some radius
$\delta$.  While it would be possible to pick $\set{\tilde{\bx}_k}_{k=1}^m$ to
be roots from the Zernike polynomials, we instead demonstrate the effectively
runtime cost where the points $\set{\tilde{\bx}_k}_{k=1}^m$ are simply chosen to
be a very dense evenly spaced grid on $B_r([1/2, 1/2])$ and the adaptive
cross-approximation (ACA) \cite{bebendorf2000}, a greedy pivoting algorithm
often used in numerical PDE methods, is used for partial rank-revealing
factorizations for $\tilde{\bm{C}}$.
\begin{figure}
  \centering
\begingroup
  \makeatletter
  \providecommand\color[2][]{%
    \GenericError{(gnuplot) \space\space\space\@spaces}{%
      Package color not loaded in conjunction with
      terminal option `colourtext'%
    }{See the gnuplot documentation for explanation.%
    }{Either use 'blacktext' in gnuplot or load the package
      color.sty in LaTeX.}%
    \renewcommand\color[2][]{}%
  }%
  \providecommand\includegraphics[2][]{%
    \GenericError{(gnuplot) \space\space\space\@spaces}{%
      Package graphicx or graphics not loaded%
    }{See the gnuplot documentation for explanation.%
    }{The gnuplot epslatex terminal needs graphicx.sty or graphics.sty.}%
    \renewcommand\includegraphics[2][]{}%
  }%
  \providecommand\rotatebox[2]{#2}%
  \@ifundefined{ifGPcolor}{%
    \newif\ifGPcolor
    \GPcolortrue
  }{}%
  \@ifundefined{ifGPblacktext}{%
    \newif\ifGPblacktext
    \GPblacktexttrue
  }{}%
  \let\gplgaddtomacro\g@addto@macro
  \gdef\gplbacktext{}%
  \gdef\gplfronttext{}%
  \makeatother
  \ifGPblacktext
    \def\colorrgb#1{}%
    \def\colorgray#1{}%
  \else
    \ifGPcolor
      \def\colorrgb#1{\color[rgb]{#1}}%
      \def\colorgray#1{\color[gray]{#1}}%
      \expandafter\def\csname LTw\endcsname{\color{white}}%
      \expandafter\def\csname LTb\endcsname{\color{black}}%
      \expandafter\def\csname LTa\endcsname{\color{black}}%
      \expandafter\def\csname LT0\endcsname{\color[rgb]{1,0,0}}%
      \expandafter\def\csname LT1\endcsname{\color[rgb]{0,1,0}}%
      \expandafter\def\csname LT2\endcsname{\color[rgb]{0,0,1}}%
      \expandafter\def\csname LT3\endcsname{\color[rgb]{1,0,1}}%
      \expandafter\def\csname LT4\endcsname{\color[rgb]{0,1,1}}%
      \expandafter\def\csname LT5\endcsname{\color[rgb]{1,1,0}}%
      \expandafter\def\csname LT6\endcsname{\color[rgb]{0,0,0}}%
      \expandafter\def\csname LT7\endcsname{\color[rgb]{1,0.3,0}}%
      \expandafter\def\csname LT8\endcsname{\color[rgb]{0.5,0.5,0.5}}%
    \else
      \def\colorrgb#1{\color{black}}%
      \def\colorgray#1{\color[gray]{#1}}%
      \expandafter\def\csname LTw\endcsname{\color{white}}%
      \expandafter\def\csname LTb\endcsname{\color{black}}%
      \expandafter\def\csname LTa\endcsname{\color{black}}%
      \expandafter\def\csname LT0\endcsname{\color{black}}%
      \expandafter\def\csname LT1\endcsname{\color{black}}%
      \expandafter\def\csname LT2\endcsname{\color{black}}%
      \expandafter\def\csname LT3\endcsname{\color{black}}%
      \expandafter\def\csname LT4\endcsname{\color{black}}%
      \expandafter\def\csname LT5\endcsname{\color{black}}%
      \expandafter\def\csname LT6\endcsname{\color{black}}%
      \expandafter\def\csname LT7\endcsname{\color{black}}%
      \expandafter\def\csname LT8\endcsname{\color{black}}%
    \fi
  \fi
    \setlength{\unitlength}{0.0500bp}%
    \ifx\gptboxheight\undefined%
      \newlength{\gptboxheight}%
      \newlength{\gptboxwidth}%
      \newsavebox{\gptboxtext}%
    \fi%
    \setlength{\fboxrule}{0.5pt}%
    \setlength{\fboxsep}{1pt}%
    \definecolor{tbcol}{rgb}{1,1,1}%
\begin{picture}(9060.00,2820.00)%
    \gplgaddtomacro\gplbacktext{%
      \csname LTb\endcsname
      \put(1707,700){\makebox(0,0)[r]{\strut{}\footnotesize $10^{-1}$}}%
      \csname LTb\endcsname
      \put(1707,1273){\makebox(0,0)[r]{\strut{}\footnotesize $10^{0}$}}%
      \csname LTb\endcsname
      \put(1707,1846){\makebox(0,0)[r]{\strut{}\footnotesize $10^{1}$}}%
      \csname LTb\endcsname
      \put(1707,2419){\makebox(0,0)[r]{\strut{}\footnotesize $10^{2}$}}%
      \csname LTb\endcsname
      \put(1945,460){\makebox(0,0){\strut{}\footnotesize $2^{10}$}}%
      \csname LTb\endcsname
      \put(2629,460){\makebox(0,0){\strut{}\footnotesize $2^{15}$}}%
      \csname LTb\endcsname
      \put(3314,460){\makebox(0,0){\strut{}\footnotesize $2^{20}$}}%
    }%
    \gplgaddtomacro\gplfronttext{%
      \csname LTb\endcsname
      \put(1058,1609){\rotatebox{-270.00}{\makebox(0,0){\strut{}\small runtime (s)}}}%
      \csname LTb\endcsname
      \put(2561,100){\makebox(0,0){\strut{}\small data size $n$}}%
      \csname LTb\endcsname
      \put(2561,2639){\makebox(0,0){\strut{}\small $\delta=0.1$}}%
    }%
    \gplgaddtomacro\gplbacktext{%
      \csname LTb\endcsname
      \put(3903,460){\makebox(0,0){\strut{}\footnotesize $2^{10}$}}%
      \csname LTb\endcsname
      \put(4588,460){\makebox(0,0){\strut{}\footnotesize $2^{15}$}}%
      \csname LTb\endcsname
      \put(5273,460){\makebox(0,0){\strut{}\footnotesize $2^{20}$}}%
    }%
    \gplgaddtomacro\gplfronttext{%
      \csname LTb\endcsname
      \put(4519,100){\makebox(0,0){\strut{}\small data size $n$}}%
      \csname LTb\endcsname
      \put(4519,2639){\makebox(0,0){\strut{}\small $\delta=0.2$}}%
    }%
    \gplgaddtomacro\gplbacktext{%
      \csname LTb\endcsname
      \put(5862,460){\makebox(0,0){\strut{}\footnotesize $2^{10}$}}%
      \csname LTb\endcsname
      \put(6547,460){\makebox(0,0){\strut{}\footnotesize $2^{15}$}}%
      \csname LTb\endcsname
      \put(7231,460){\makebox(0,0){\strut{}\footnotesize $2^{20}$}}%
    }%
    \gplgaddtomacro\gplfronttext{%
      \csname LTb\endcsname
      \put(8339,2400){\makebox(0,0)[r]{\strut{}\tiny $\rho=0.1$}}%
      \csname LTb\endcsname
      \put(8339,2160){\makebox(0,0)[r]{\strut{}\tiny $\rho=1$}}%
      \csname LTb\endcsname
      \put(8339,1920){\makebox(0,0)[r]{\strut{}\tiny $\rho=10$}}%
      \csname LTb\endcsname
      \put(8339,1680){\makebox(0,0)[r]{\strut{}\tiny $\bO(n)$}}%
      \csname LTb\endcsname
      \put(6478,100){\makebox(0,0){\strut{}\small data size $n$}}%
      \csname LTb\endcsname
      \put(6478,2639){\makebox(0,0){\strut{}\small $\delta=0.3$}}%
    }%
    \gplbacktext
    \put(0,0){\includegraphics[width={453.00bp},height={141.00bp}]{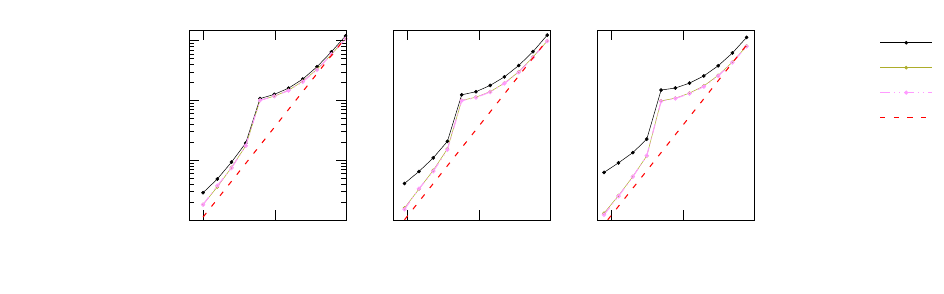}}%
    \gplfronttext
  \end{picture}%
\endgroup

  \caption{The runtime cost (in seconds) of assembling a full basis $\bm{Q}$ for
  the non-degenerate singular vector space of $\tilde{\bm{\Lambda}} = \bS^{-1}
  \tilde{\bm{C}}$ for a growing number of points chosen uniformly on $[0,1]^2
  \setminus B_{\delta}([1/2,1/2])$ using the Gaussian covariance function
  (\ref{eq:gausskernel}) with various values of the range parameter $\rho$
  (indicated by color) and $\tau^2$ fixed at $0.1$.}
  \label{fig:runtime_gauss}
\end{figure}

Figure \ref{fig:runtime_gauss} shows the runtime cost of assembling $\bm{Q}$ in
the setting where (i) one can directly assemble $\bm{G}$ such that $\bS \approx
\bm{G} \bm{G}^T + \tau^2 \I$ and (ii) one can directly assemble $\tilde{\bm{C}}
\approx \bm{U}_c \bm{V}_c^T$. Several patterns are worth commenting on. First,
we note that the similar runtime cost between different radius values $\delta$
is in part due to the fact that a higher fraction of the $n$ points are being
withheld as missing as $\delta$ increases (the top data size for $\delta=0.3$,
for example, was $n=720\,652$, compared to $n\approx 1.01\times
10^6$ for $\delta=0.1$). Second, we note that the range parameter $\rho$
decreasing leads to weaker dependence across the domain, and the impact that
this has on the rank structure of $\bS$ is subtle.  On the one hand, less
dependence means more local features, which raises the rank of $\bS$ and
$\tilde{\bm{C}}$. On the other hand, however, less dependence also means that
far-away points are less informative.  This duality is reminiscent of questions
pertaining to relative versus absolute error or efficiency, and different
stopping criteria in sketching-based partial factorizations may make a shrinking
$\rho$ more or less costly relative to a larger one.

As a second observation, we remind the reader that the runtime cost scaling here
is dependent on the asymptotic sampling regimes. As measurements are made more
and more densely in $[0,1]^2$, for example, the number of columns in $\bm{Q}$
grows sub-linearly in $n$ and eventually stops growing at all, as the function
$\bx_0 \mapsto \bm{\lambda}(\bx_0)$ is being approximated on a fixed region by
points that are also growing more dense in a fixed region (a fact that has deep
implications and connections to parameter identifiability theory in the spatial
statistics literature, see \cite{stein1999} and references within for more
information).  If the domain is growing, an alternative method for rapid solves
with $\bS^{-1}$ would be necessary (potentially using the \emph{fast Gauss
transform} \cite{greengard1991} and a preconditioned conjugate gradient approach
for linear solves). For more complex regions and asymptotic regimes or
covariance functions, the runtime cost of this procedure may be more subtle. 

\subsection{Rosenbrock function interpolation with noisy data}
\label{sec:app_rosenbrock}

Next, we consider an application setting where one has observed data
\begin{equation} \label{eq:data}
  y(\bx) = f(\bx + (1,1)) + \eps_j,
\end{equation}
where $\eps_j \sim \Nd(0, \tau^2)$ and $f(\bx) = \sqrt{(a - x_1)^2 + b(x_2 -
x_1^2)^2}$ is the square root of the \emph{Rosenbrock} function chosen with
standard parameters $(a,b) = (1,100)$, a standard test function in the
optimization and function approximation literature, that has been shifted so
that its minimizer is at the origin. In particular, we consider the problem in
which data have been measured on the domain $[-2,2]^2$ except for a missing
region of radius $\delta = 0.75$ centered at the origin. The goal is to rapidly
obtain accurate interpolants for this function in the missing region, with
potential downstream applications such as Bayesian optimization in which
practitioners actually attempt to optimize over such missing regions. As before,
a Gaussian covariance function is used, although this time with a geometric
anisotropy estimated from data, so that the fitted covariance function is
\begin{equation*} 
  K(\bx - \bx')
  =
  \sigma^2
  \exp \left(
    - (\bx - \bx')^T \bm{M} (\bx - \bx')
  \right)
\end{equation*}
where we note the model-implied marginal variance is $\sigma^2 \approx 2325$ and
the anisotropy is given by
\begin{equation*} 
  \bm{M} = 
  \mat{
    \cos(\theta) & -\sin(\theta) \\
    \sin(\theta) & \cos(\theta)
  }
  \mat{
  0.021 & 0 \\
  0 & 0.26
  }
  \mat{
    \cos(\theta) & -\sin(\theta) \\
    \sin(\theta) & \cos(\theta)
  }^T,
\end{equation*}
with $\theta \approx 1.374$. The data used for this prediction problem is shown
in Figure \ref{fig:rosenbrock_data}.
\begin{figure}[!ht]
  \centering
\begingroup
  \makeatletter
  \providecommand\color[2][]{%
    \GenericError{(gnuplot) \space\space\space\@spaces}{%
      Package color not loaded in conjunction with
      terminal option `colourtext'%
    }{See the gnuplot documentation for explanation.%
    }{Either use 'blacktext' in gnuplot or load the package
      color.sty in LaTeX.}%
    \renewcommand\color[2][]{}%
  }%
  \providecommand\includegraphics[2][]{%
    \GenericError{(gnuplot) \space\space\space\@spaces}{%
      Package graphicx or graphics not loaded%
    }{See the gnuplot documentation for explanation.%
    }{The gnuplot epslatex terminal needs graphicx.sty or graphics.sty.}%
    \renewcommand\includegraphics[2][]{}%
  }%
  \providecommand\rotatebox[2]{#2}%
  \@ifundefined{ifGPcolor}{%
    \newif\ifGPcolor
    \GPcolortrue
  }{}%
  \@ifundefined{ifGPblacktext}{%
    \newif\ifGPblacktext
    \GPblacktexttrue
  }{}%
  \let\gplgaddtomacro\g@addto@macro
  \gdef\gplbacktext{}%
  \gdef\gplfronttext{}%
  \makeatother
  \ifGPblacktext
    \def\colorrgb#1{}%
    \def\colorgray#1{}%
  \else
    \ifGPcolor
      \def\colorrgb#1{\color[rgb]{#1}}%
      \def\colorgray#1{\color[gray]{#1}}%
      \expandafter\def\csname LTw\endcsname{\color{white}}%
      \expandafter\def\csname LTb\endcsname{\color{black}}%
      \expandafter\def\csname LTa\endcsname{\color{black}}%
      \expandafter\def\csname LT0\endcsname{\color[rgb]{1,0,0}}%
      \expandafter\def\csname LT1\endcsname{\color[rgb]{0,1,0}}%
      \expandafter\def\csname LT2\endcsname{\color[rgb]{0,0,1}}%
      \expandafter\def\csname LT3\endcsname{\color[rgb]{1,0,1}}%
      \expandafter\def\csname LT4\endcsname{\color[rgb]{0,1,1}}%
      \expandafter\def\csname LT5\endcsname{\color[rgb]{1,1,0}}%
      \expandafter\def\csname LT6\endcsname{\color[rgb]{0,0,0}}%
      \expandafter\def\csname LT7\endcsname{\color[rgb]{1,0.3,0}}%
      \expandafter\def\csname LT8\endcsname{\color[rgb]{0.5,0.5,0.5}}%
    \else
      \def\colorrgb#1{\color{black}}%
      \def\colorgray#1{\color[gray]{#1}}%
      \expandafter\def\csname LTw\endcsname{\color{white}}%
      \expandafter\def\csname LTb\endcsname{\color{black}}%
      \expandafter\def\csname LTa\endcsname{\color{black}}%
      \expandafter\def\csname LT0\endcsname{\color{black}}%
      \expandafter\def\csname LT1\endcsname{\color{black}}%
      \expandafter\def\csname LT2\endcsname{\color{black}}%
      \expandafter\def\csname LT3\endcsname{\color{black}}%
      \expandafter\def\csname LT4\endcsname{\color{black}}%
      \expandafter\def\csname LT5\endcsname{\color{black}}%
      \expandafter\def\csname LT6\endcsname{\color{black}}%
      \expandafter\def\csname LT7\endcsname{\color{black}}%
      \expandafter\def\csname LT8\endcsname{\color{black}}%
    \fi
  \fi
    \setlength{\unitlength}{0.0500bp}%
    \ifx\gptboxheight\undefined%
      \newlength{\gptboxheight}%
      \newlength{\gptboxwidth}%
      \newsavebox{\gptboxtext}%
    \fi%
    \setlength{\fboxrule}{0.5pt}%
    \setlength{\fboxsep}{1pt}%
    \definecolor{tbcol}{rgb}{1,1,1}%
\begin{picture}(3400.00,2820.00)%
    \gplgaddtomacro\gplbacktext{%
      \csname LTb\endcsname
      \put(68,560){\makebox(0,0)[r]{\strut{}\footnotesize -1.5}}%
      \csname LTb\endcsname
      \put(68,1120){\makebox(0,0)[r]{\strut{}\footnotesize -0.5}}%
      \csname LTb\endcsname
      \put(68,1679){\makebox(0,0)[r]{\strut{}\footnotesize 0.5}}%
      \csname LTb\endcsname
      \put(68,2239){\makebox(0,0)[r]{\strut{}\footnotesize 1.5}}%
      \csname LTb\endcsname
      \put(528,40){\makebox(0,0){\strut{}\footnotesize -1.5}}%
      \csname LTb\endcsname
      \put(1246,40){\makebox(0,0){\strut{}\footnotesize -0.5}}%
      \csname LTb\endcsname
      \put(1964,40){\makebox(0,0){\strut{}\footnotesize 0.5}}%
      \csname LTb\endcsname
      \put(2682,40){\makebox(0,0){\strut{}\footnotesize 1.5}}%
    }%
    \gplgaddtomacro\gplfronttext{%
      \csname LTb\endcsname
      \put(3358,483){\makebox(0,0)[l]{\strut{}\tiny 0}}%
      \csname LTb\endcsname
      \put(3358,890){\makebox(0,0)[l]{\strut{}\tiny 20}}%
      \csname LTb\endcsname
      \put(3358,1298){\makebox(0,0)[l]{\strut{}\tiny 40}}%
      \csname LTb\endcsname
      \put(3358,1705){\makebox(0,0)[l]{\strut{}\tiny 60}}%
      \csname LTb\endcsname
      \put(3358,2112){\makebox(0,0)[l]{\strut{}\tiny 80}}%
      \csname LTb\endcsname
      \put(3358,2519){\makebox(0,0)[l]{\strut{}\tiny 100}}%
    }%
    \gplbacktext
    \put(0,0){\includegraphics[width={170.00bp},height={141.00bp}]{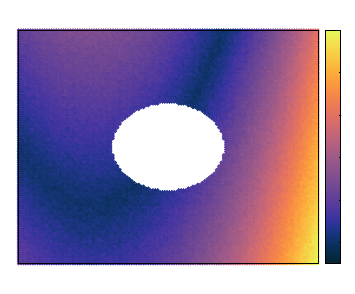}}%
    \gplfronttext
  \end{picture}%
\endgroup

  \caption{The dataset used in the prediction problem analyzed in Figure
  \ref{fig:rosenbrock_highnoise} below. The data size is $n \approx 15\,000$,
  and measurements are given as in (\ref{eq:data}).}
  \label{fig:rosenbrock_data}
\end{figure}
As discussed above, measurement noise significantly weakens the locality of
Kriging prediction weights, and for an analytic process like the underlying
$f(\bx)$ here the impact of noise on prediction results can be significant even
when $f(\bx)$ itself has short dependence lengthscales.
\begin{figure}[!ht]
  \centering
\begingroup
  \makeatletter
  \providecommand\color[2][]{%
    \GenericError{(gnuplot) \space\space\space\@spaces}{%
      Package color not loaded in conjunction with
      terminal option `colourtext'%
    }{See the gnuplot documentation for explanation.%
    }{Either use 'blacktext' in gnuplot or load the package
      color.sty in LaTeX.}%
    \renewcommand\color[2][]{}%
  }%
  \providecommand\includegraphics[2][]{%
    \GenericError{(gnuplot) \space\space\space\@spaces}{%
      Package graphicx or graphics not loaded%
    }{See the gnuplot documentation for explanation.%
    }{The gnuplot epslatex terminal needs graphicx.sty or graphics.sty.}%
    \renewcommand\includegraphics[2][]{}%
  }%
  \providecommand\rotatebox[2]{#2}%
  \@ifundefined{ifGPcolor}{%
    \newif\ifGPcolor
    \GPcolortrue
  }{}%
  \@ifundefined{ifGPblacktext}{%
    \newif\ifGPblacktext
    \GPblacktexttrue
  }{}%
  \let\gplgaddtomacro\g@addto@macro
  \gdef\gplbacktext{}%
  \gdef\gplfronttext{}%
  \makeatother
  \ifGPblacktext
    \def\colorrgb#1{}%
    \def\colorgray#1{}%
  \else
    \ifGPcolor
      \def\colorrgb#1{\color[rgb]{#1}}%
      \def\colorgray#1{\color[gray]{#1}}%
      \expandafter\def\csname LTw\endcsname{\color{white}}%
      \expandafter\def\csname LTb\endcsname{\color{black}}%
      \expandafter\def\csname LTa\endcsname{\color{black}}%
      \expandafter\def\csname LT0\endcsname{\color[rgb]{1,0,0}}%
      \expandafter\def\csname LT1\endcsname{\color[rgb]{0,1,0}}%
      \expandafter\def\csname LT2\endcsname{\color[rgb]{0,0,1}}%
      \expandafter\def\csname LT3\endcsname{\color[rgb]{1,0,1}}%
      \expandafter\def\csname LT4\endcsname{\color[rgb]{0,1,1}}%
      \expandafter\def\csname LT5\endcsname{\color[rgb]{1,1,0}}%
      \expandafter\def\csname LT6\endcsname{\color[rgb]{0,0,0}}%
      \expandafter\def\csname LT7\endcsname{\color[rgb]{1,0.3,0}}%
      \expandafter\def\csname LT8\endcsname{\color[rgb]{0.5,0.5,0.5}}%
    \else
      \def\colorrgb#1{\color{black}}%
      \def\colorgray#1{\color[gray]{#1}}%
      \expandafter\def\csname LTw\endcsname{\color{white}}%
      \expandafter\def\csname LTb\endcsname{\color{black}}%
      \expandafter\def\csname LTa\endcsname{\color{black}}%
      \expandafter\def\csname LT0\endcsname{\color{black}}%
      \expandafter\def\csname LT1\endcsname{\color{black}}%
      \expandafter\def\csname LT2\endcsname{\color{black}}%
      \expandafter\def\csname LT3\endcsname{\color{black}}%
      \expandafter\def\csname LT4\endcsname{\color{black}}%
      \expandafter\def\csname LT5\endcsname{\color{black}}%
      \expandafter\def\csname LT6\endcsname{\color{black}}%
      \expandafter\def\csname LT7\endcsname{\color{black}}%
      \expandafter\def\csname LT8\endcsname{\color{black}}%
    \fi
  \fi
    \setlength{\unitlength}{0.0500bp}%
    \ifx\gptboxheight\undefined%
      \newlength{\gptboxheight}%
      \newlength{\gptboxwidth}%
      \newsavebox{\gptboxtext}%
    \fi%
    \setlength{\fboxrule}{0.5pt}%
    \setlength{\fboxsep}{1pt}%
    \definecolor{tbcol}{rgb}{1,1,1}%
\begin{picture}(7920.00,3400.00)%
    \gplgaddtomacro\gplbacktext{%
    }%
    \gplgaddtomacro\gplfronttext{%
      \csname LTb\endcsname
      \put(199,2344){\rotatebox{-270.00}{\makebox(0,0){\strut{}\footnotesize Cond. $\mu$}}}%
      \csname LTb\endcsname
      \put(1758,1647){\makebox(0,0)[l]{\strut{}\tiny -2}}%
      \csname LTb\endcsname
      \put(1758,2344){\makebox(0,0)[l]{\strut{}\tiny 10}}%
      \csname LTb\endcsname
      \put(1758,2925){\makebox(0,0)[l]{\strut{}\tiny 20}}%
      \csname LTb\endcsname
      \put(982,3161){\makebox(0,0){\strut{}\footnotesize KNN: $k=100$}}%
    }%
    \gplgaddtomacro\gplbacktext{%
    }%
    \gplgaddtomacro\gplfronttext{%
      \csname LTb\endcsname
      \put(3605,1885){\makebox(0,0)[l]{\strut{}\tiny 5}}%
      \csname LTb\endcsname
      \put(3605,2225){\makebox(0,0)[l]{\strut{}\tiny 10}}%
      \csname LTb\endcsname
      \put(3605,2565){\makebox(0,0)[l]{\strut{}\tiny 15}}%
      \csname LTb\endcsname
      \put(3605,2905){\makebox(0,0)[l]{\strut{}\tiny 20}}%
      \csname LTb\endcsname
      \put(2829,3161){\makebox(0,0){\strut{}\footnotesize KNN: $k=1000$}}%
    }%
    \gplgaddtomacro\gplbacktext{%
    }%
    \gplgaddtomacro\gplfronttext{%
      \csname LTb\endcsname
      \put(5452,1885){\makebox(0,0)[l]{\strut{}\tiny 5}}%
      \csname LTb\endcsname
      \put(5452,2225){\makebox(0,0)[l]{\strut{}\tiny 10}}%
      \csname LTb\endcsname
      \put(5452,2565){\makebox(0,0)[l]{\strut{}\tiny 15}}%
      \csname LTb\endcsname
      \put(5452,2905){\makebox(0,0)[l]{\strut{}\tiny 20}}%
      \csname LTb\endcsname
      \put(4675,3161){\makebox(0,0){\strut{}\footnotesize $\bm{Q} \in \R^{n \times 30}$}}%
    }%
    \gplgaddtomacro\gplbacktext{%
    }%
    \gplgaddtomacro\gplfronttext{%
      \csname LTb\endcsname
      \put(7298,1885){\makebox(0,0)[l]{\strut{}\tiny 5}}%
      \csname LTb\endcsname
      \put(7298,2225){\makebox(0,0)[l]{\strut{}\tiny 10}}%
      \csname LTb\endcsname
      \put(7298,2565){\makebox(0,0)[l]{\strut{}\tiny 15}}%
      \csname LTb\endcsname
      \put(7298,2905){\makebox(0,0)[l]{\strut{}\tiny 20}}%
      \csname LTb\endcsname
      \put(6522,3161){\makebox(0,0){\strut{}\footnotesize exact}}%
    }%
    \gplgaddtomacro\gplbacktext{%
    }%
    \gplgaddtomacro\gplfronttext{%
      \csname LTb\endcsname
      \put(199,781){\rotatebox{-270.00}{\makebox(0,0){\strut{}\footnotesize Cond. $\sigma$}}}%
      \csname LTb\endcsname
      \put(1758,84){\makebox(0,0)[l]{\strut{}\tiny 0}}%
      \csname LTb\endcsname
      \put(1758,477){\makebox(0,0)[l]{\strut{}\tiny 1}}%
      \csname LTb\endcsname
      \put(1758,869){\makebox(0,0)[l]{\strut{}\tiny 2}}%
      \csname LTb\endcsname
      \put(1758,1262){\makebox(0,0)[l]{\strut{}\tiny 3}}%
    }%
    \gplgaddtomacro\gplbacktext{%
    }%
    \gplgaddtomacro\gplfronttext{%
      \csname LTb\endcsname
      \put(3605,84){\makebox(0,0)[l]{\strut{}\tiny 0}}%
      \csname LTb\endcsname
      \put(3605,781){\makebox(0,0)[l]{\strut{}\tiny 0.25}}%
      \csname LTb\endcsname
      \put(3605,1478){\makebox(0,0)[l]{\strut{}\tiny 0.5}}%
    }%
    \gplgaddtomacro\gplbacktext{%
    }%
    \gplgaddtomacro\gplfronttext{%
      \csname LTb\endcsname
      \put(5452,84){\makebox(0,0)[l]{\strut{}\tiny 0.05}}%
      \csname LTb\endcsname
      \put(5452,642){\makebox(0,0)[l]{\strut{}\tiny 0.06}}%
      \csname LTb\endcsname
      \put(5452,1199){\makebox(0,0)[l]{\strut{}\tiny 0.07}}%
    }%
    \gplgaddtomacro\gplbacktext{%
    }%
    \gplgaddtomacro\gplfronttext{%
      \csname LTb\endcsname
      \put(7298,309){\makebox(0,0)[l]{\strut{}\tiny 0.003}}%
      \csname LTb\endcsname
      \put(7298,759){\makebox(0,0)[l]{\strut{}\tiny 0.004}}%
      \csname LTb\endcsname
      \put(7298,1208){\makebox(0,0)[l]{\strut{}\tiny 0.005}}%
    }%
    \gplbacktext
    \put(0,0){\includegraphics[width={396.00bp},height={170.00bp}]{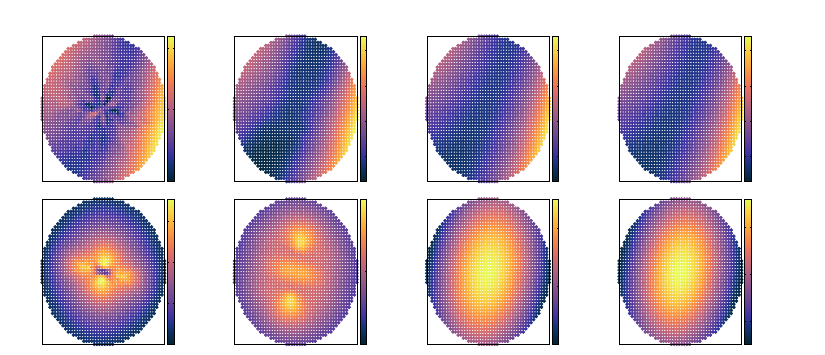}}%
    \gplfronttext
  \end{picture}%
\endgroup

  \caption{A visual summary of the conditional mean and variance for the
  Rosenbrock function prediction problem using $k$-nearest neighbors and our
  dense $\bm{Q}^T \by$ conditioning strategy, demonstrating the instability of
  nearest-neighbor prediction for this problem in the case of a mild noise.}
  \label{fig:rosenbrock_highnoise}
\end{figure}
Figure \ref{fig:rosenbrock_highnoise} shows a striking demonstration of this
phenomenon: with measurement noise $\tau^2 = 2$, representing a reasonably high signal to
noise ratio (SNR), predicting with even $k=1\,000$ nearest neighbors produces
very strong artifacts, particularly in the conditional variance. Predicting with
a fixed $r=30$ columns $\bm{Q}$ computed from the dominant singular vectors of
$\tilde{\bm{\Lambda}} = \bS^{-1} \tilde{\bm{C}}$, on the other hand, gives
predictions and conditional variances that are nearly identical in size and
structure to that of the exact conditional distribution given all the data. We
note that the magnitude of the conditional variance given $\bm{Q}^T \by$ is
materially larger than it is given the full data $\by$. But this is arguably a
demonstration that one does not \emph{need} for $\bm{Q}$ to have as many columns
as $\tilde{\bm{\Lambda}}$ has non-degenerate singular vectors for the posterior
structure of the prediction problem to be very faithfully reproduced. In
Bayesian optimization settings, for example, one frequently attempts to maximize
$\bx \mapsto E[f(\bx) \sv \by] + \sqrt{V[f(\bx) \sv \by]}$. While
replacing $\by$ with the $k$ nearest neighbors would have a very significant
impact on that objective function and where it is minimized (as demonstrated by
Figure \ref{fig:rosenbrock_highnoise}), approximating those conditionals using
$\bm{Q}^T \by$ instead of $\by$ will have effectively no impact, and will still
provide an objective function that can be evaluated in $\bO(1)$ runtime cost
using the anterpolation recompression described in Section \ref{sec:online}.

\subsection{Guaranteed definiteness with off-diagonal block compression} \label{sec:odlr}

A very popular method in the applied math community for working with large
structured matrices is to obtain low rank approximations for off-diagonal
blocks. Concretely, in the case of a GP model, one may consider a collection of
measurements $\by = [f(\bx_j)]_{j=1}^n$, where $f$ is a mean-zero GP whose
kernel is smooth away from the origin, and partition $\by$ into $\by = [\by_1,
\by_2]$ based on, e.g., splitting the domain. Since the kernel for $f$ is smooth
away from the origin, the arguments of Section \ref{sec:low-rank-crosscov} would
imply that the covariance matrix may be approximated as
\begin{equation} \label{eq:odlr}
  \V
  \mat{\by_1 \\ \by_2}  
  =
  \mat{\bS_{11}   & \bS_{12} \\
       \bS_{12}^T & \bS_{22}
  }
  \approx
  \mat{\bS_{11}   & \bm{U} \bm{V}^T \\
       \bm{V} \bm{U}^T & \bS_{22}
  }.
\end{equation}
Most methods for problems in this category use direct algebraic methods to
factorize $\bS_{12} \approx \bm{U} \bm{V}^T$, such as the ACA
\cite{bebendorf2000}. If one computes $\bm{U}$ and $\bm{V}$ so that this
approximation is accurate to nearly machine precision, then the resulting full
matrix approximation (the right-hand side of (\ref{eq:odlr})) will remain
positive definite. A fundamental challenge with this category of approximation,
however, is that the numerical rank of a cross-covariance like $\bS_{12}$ will
change based on point geometry, kernel properties and parameters, dimension, and
so on.  As a result, methods that cap the permitted rank implied by $\bm{U}
\bm{V}^T$, which is often necessary for computational complexity guarantees, can
run the risk of producing less accurate overall approximations to the true
matrix in (\ref{eq:odlr}), and oftentimes even produce indefinite approximations
to definite matrices \cite{geoga2020}. For GP applications in particular that
require a log-determinant, this can be problematic.

The method introduced here offers a natural alternative for producing these
low-rank approximations to off-diagonal blocks by instead compressing them in
the inverse square-root domain. Letting $\bm{C}_{11}$ and $\bm{C}_{22|1}$ denote 
the lower Cholesky factors of $\bS_{11}$ and the Schur complement 
$\bS_{22} - \bS_{12}^T \bS_{11}^{-1} \bS_{12}$ respectively (so that 
$\bS_{11} = \bm{C}_{11} \bm{C}_{11}^T$), we recall that the block Cholesky 
factorization for the precision matrix of the data $[\by_1, \by_2]$ is given by
\begin{equation} \label{eq:ichol}
  \mat{\bS_{11}   & \bS_{12} \\
       \bS_{12}^T & \bS_{22}
  }^{-1}
  =
  \bm{L}^T \bm{L}
  \quad \text{where} \quad
  \bm{L} = 
  \mat{
    \bm{C}_{11}^{-1} & \bm{0} \\
    -\bm{C}_{22|1}^{-1} \bS_{12}^T \bS_{11}^{-1} & \bm{C}_{22|1}^{-1}
  }.
\end{equation}
With that in mind, if we obtain a $\bm{Q} \in \R^{n_1 \times r}$ such that
$E[\by_2 \sv \by_1] \approx E[\by_2 \sv \bm{Q}^T \by_1]$, then we may build a
block-sparse approximation to $\bm{L}$. Defining $\bS_{QQ} = \bm{Q}^T \bS_{11} \bm{Q}$ 
and letting the remaining $\bm{C}$ matrices denote the lower Cholesky factors of their 
respective Schur complements (so that, as an example, $\bm{C}_{Q_\perp | Q}$ is the 
lower Cholesky factor of $\bm{Q}_\perp^T \bS_{11} \bm{Q}_\perp - \bm{Q}_\perp^T
\bS_{11} \bm{Q} \bS_{QQ}^{-1} \bm{Q}^T \bS_{11} \bm{Q}_\perp$), we have
\begin{equation*} \label{eq:ichol_appx}
  \set{
  \V \mat{
    \bm{Q}^T         \by_1 \\
    \bm{Q}_{\perp}^T \by_1 \\
    \by_2
  }
  }^{-1}
  \approx 
  \tilde{\bm{L}}^T \tilde{\bm{L}}
  \quad 
  \text{where} \quad
  \tilde{\bm{L}}
  =
  \mat{
  \bm{C}_{QQ}^{-1} & \bm{0} & \bm{0} \\
  -\bm{C}_{Q_\perp | Q}^{-1} \bm{Q}_\perp^T \bS_{11} \bm{Q} \bS_{QQ}^{-1} & \bm{C}_{Q_\perp | Q}^{-1} & \bm{0} \\
  -\bm{C}_{22|Q}^{-1} \bS_{12}^T \bm{Q} \bS_{QQ}^{-1} & \bm{0} & \bm{C}_{22|Q}^{-1}
  }.
\end{equation*}
By virtue of effectively incorporating this low rank approximation in square
root space and directly forming a valid Cholesky factor, the implied
approximation to the full matrix (which is a permuted form of $(\tilde{\bm{L}}^T
\tilde{\bm{L}})^{-1}$) will naturally be positive definite. And so unlike in the
setting where one directly approximates $\bS_{12}$, this approach makes fixed
rank and non-adaptive matrix compression methods safe. Figure \ref{fig:odlr}
shows a demonstration of KL-accuracy of the two methods in two settings of
observing a Mat\'ern process on $[0,1]^2$. In the first setting, the points
$\set{\bx_j}_{j=1}^{10^4}$ are uniformly distributed and split into $\by_1 \in
\R^{n_1}$ and $\by_2 \in \R^{n_2}$ based on where $x_1 \leq 1/2$. Since
locations for $\by_1$ and $\by_2$ can be arbitrarily close, we refer to this as
weakly separated. Alternatively, we also consider the setting where points for
$\by_1$ are drawn randomly on $[0, 0.475] \times [0,1]$ and points for $\by_2$
are drawn on $[0.525, 1] \times [0,1]$. Since these points are at least a normed
distance of $0.05$ away, we call this strongly separated, as
$\text{dist}(\mathcal{D}, \Dpred) \geq 0.05$ in the terminology of Section
\ref{sec:conditioning}.
\begin{figure}[!ht]
  \centering
\begingroup
  \makeatletter
  \providecommand\color[2][]{%
    \GenericError{(gnuplot) \space\space\space\@spaces}{%
      Package color not loaded in conjunction with
      terminal option `colourtext'%
    }{See the gnuplot documentation for explanation.%
    }{Either use 'blacktext' in gnuplot or load the package
      color.sty in LaTeX.}%
    \renewcommand\color[2][]{}%
  }%
  \providecommand\includegraphics[2][]{%
    \GenericError{(gnuplot) \space\space\space\@spaces}{%
      Package graphicx or graphics not loaded%
    }{See the gnuplot documentation for explanation.%
    }{The gnuplot epslatex terminal needs graphicx.sty or graphics.sty.}%
    \renewcommand\includegraphics[2][]{}%
  }%
  \providecommand\rotatebox[2]{#2}%
  \@ifundefined{ifGPcolor}{%
    \newif\ifGPcolor
    \GPcolortrue
  }{}%
  \@ifundefined{ifGPblacktext}{%
    \newif\ifGPblacktext
    \GPblacktexttrue
  }{}%
  \let\gplgaddtomacro\g@addto@macro
  \gdef\gplbacktext{}%
  \gdef\gplfronttext{}%
  \makeatother
  \ifGPblacktext
    \def\colorrgb#1{}%
    \def\colorgray#1{}%
  \else
    \ifGPcolor
      \def\colorrgb#1{\color[rgb]{#1}}%
      \def\colorgray#1{\color[gray]{#1}}%
      \expandafter\def\csname LTw\endcsname{\color{white}}%
      \expandafter\def\csname LTb\endcsname{\color{black}}%
      \expandafter\def\csname LTa\endcsname{\color{black}}%
      \expandafter\def\csname LT0\endcsname{\color[rgb]{1,0,0}}%
      \expandafter\def\csname LT1\endcsname{\color[rgb]{0,1,0}}%
      \expandafter\def\csname LT2\endcsname{\color[rgb]{0,0,1}}%
      \expandafter\def\csname LT3\endcsname{\color[rgb]{1,0,1}}%
      \expandafter\def\csname LT4\endcsname{\color[rgb]{0,1,1}}%
      \expandafter\def\csname LT5\endcsname{\color[rgb]{1,1,0}}%
      \expandafter\def\csname LT6\endcsname{\color[rgb]{0,0,0}}%
      \expandafter\def\csname LT7\endcsname{\color[rgb]{1,0.3,0}}%
      \expandafter\def\csname LT8\endcsname{\color[rgb]{0.5,0.5,0.5}}%
    \else
      \def\colorrgb#1{\color{black}}%
      \def\colorgray#1{\color[gray]{#1}}%
      \expandafter\def\csname LTw\endcsname{\color{white}}%
      \expandafter\def\csname LTb\endcsname{\color{black}}%
      \expandafter\def\csname LTa\endcsname{\color{black}}%
      \expandafter\def\csname LT0\endcsname{\color{black}}%
      \expandafter\def\csname LT1\endcsname{\color{black}}%
      \expandafter\def\csname LT2\endcsname{\color{black}}%
      \expandafter\def\csname LT3\endcsname{\color{black}}%
      \expandafter\def\csname LT4\endcsname{\color{black}}%
      \expandafter\def\csname LT5\endcsname{\color{black}}%
      \expandafter\def\csname LT6\endcsname{\color{black}}%
      \expandafter\def\csname LT7\endcsname{\color{black}}%
      \expandafter\def\csname LT8\endcsname{\color{black}}%
    \fi
  \fi
    \setlength{\unitlength}{0.0500bp}%
    \ifx\gptboxheight\undefined%
      \newlength{\gptboxheight}%
      \newlength{\gptboxwidth}%
      \newsavebox{\gptboxtext}%
    \fi%
    \setlength{\fboxrule}{0.5pt}%
    \setlength{\fboxsep}{1pt}%
    \definecolor{tbcol}{rgb}{1,1,1}%
\begin{picture}(9060.00,3400.00)%
    \gplgaddtomacro\gplbacktext{%
      \csname LTb\endcsname
      \put(803,837){\makebox(0,0)[r]{\strut{}\scriptsize $10^{-3}$}}%
      \csname LTb\endcsname
      \put(803,1196){\makebox(0,0)[r]{\strut{}\scriptsize $10^{-2}$}}%
      \csname LTb\endcsname
      \put(803,1554){\makebox(0,0)[r]{\strut{}\scriptsize $10^{-1}$}}%
      \csname LTb\endcsname
      \put(803,1913){\makebox(0,0)[r]{\strut{}\scriptsize $10^{0}$}}%
      \csname LTb\endcsname
      \put(803,2272){\makebox(0,0)[r]{\strut{}\scriptsize $10^{1}$}}%
      \csname LTb\endcsname
      \put(803,2631){\makebox(0,0)[r]{\strut{}\scriptsize $10^{2}$}}%
      \csname LTb\endcsname
      \put(1293,436){\makebox(0,0){\strut{}\scriptsize 20}}%
      \csname LTb\endcsname
      \put(1702,436){\makebox(0,0){\strut{}\scriptsize 40}}%
      \csname LTb\endcsname
      \put(2112,436){\makebox(0,0){\strut{}\scriptsize 60}}%
      \csname LTb\endcsname
      \put(2521,436){\makebox(0,0){\strut{}\scriptsize 80}}%
      \csname LTb\endcsname
      \put(2931,436){\makebox(0,0){\strut{}\scriptsize 100}}%
      \csname LTb\endcsname
      \put(3340,436){\makebox(0,0){\strut{}\scriptsize 120}}%
      \csname LTb\endcsname
      \put(3750,436){\makebox(0,0){\strut{}\scriptsize 140}}%
    }%
    \gplgaddtomacro\gplfronttext{%
      \csname LTb\endcsname
      \put(154,1774){\rotatebox{-270.00}{\makebox(0,0){\strut{}\small KL divergence}}}%
      \csname LTb\endcsname
      \put(2429,76){\makebox(0,0){\strut{}\small rank $r$}}%
      \csname LTb\endcsname
      \put(2429,3112){\makebox(0,0){\strut{}\small Weakly separated}}%
    }%
    \gplgaddtomacro\gplbacktext{%
      \csname LTb\endcsname
      \put(4984,684){\makebox(0,0)[r]{\strut{}\scriptsize $10^{-10}$}}%
      \csname LTb\endcsname
      \put(4984,1089){\makebox(0,0)[r]{\strut{}\scriptsize $10^{-8}$}}%
      \csname LTb\endcsname
      \put(4984,1494){\makebox(0,0)[r]{\strut{}\scriptsize $10^{-6}$}}%
      \csname LTb\endcsname
      \put(4984,1899){\makebox(0,0)[r]{\strut{}\scriptsize $10^{-4}$}}%
      \csname LTb\endcsname
      \put(4984,2304){\makebox(0,0)[r]{\strut{}\scriptsize $10^{-2}$}}%
      \csname LTb\endcsname
      \put(4984,2709){\makebox(0,0)[r]{\strut{}\scriptsize $10^{0}$}}%
      \csname LTb\endcsname
      \put(5474,436){\makebox(0,0){\strut{}\scriptsize 20}}%
      \csname LTb\endcsname
      \put(5883,436){\makebox(0,0){\strut{}\scriptsize 40}}%
      \csname LTb\endcsname
      \put(6293,436){\makebox(0,0){\strut{}\scriptsize 60}}%
      \csname LTb\endcsname
      \put(6702,436){\makebox(0,0){\strut{}\scriptsize 80}}%
      \csname LTb\endcsname
      \put(7112,436){\makebox(0,0){\strut{}\scriptsize 100}}%
      \csname LTb\endcsname
      \put(7521,436){\makebox(0,0){\strut{}\scriptsize 120}}%
      \csname LTb\endcsname
      \put(7931,436){\makebox(0,0){\strut{}\scriptsize 140}}%
    }%
    \gplgaddtomacro\gplfronttext{%
      \csname LTb\endcsname
      \put(7687,2657){\makebox(0,0)[r]{\strut{}\scriptsize direct $\bS_{12} \approx \bm{U} \bm{V}^T$}}%
      \csname LTb\endcsname
      \put(7687,2417){\makebox(0,0)[r]{\strut{}\scriptsize (\ref{eq:ichol_appx})-based approx.}}%
      \csname LTb\endcsname
      \put(6610,76){\makebox(0,0){\strut{}\small rank $r$}}%
      \csname LTb\endcsname
      \put(6610,3112){\makebox(0,0){\strut{}\small Strongly separated}}%
    }%
    \gplbacktext
    \put(0,0){\includegraphics[width={453.00bp},height={170.00bp}]{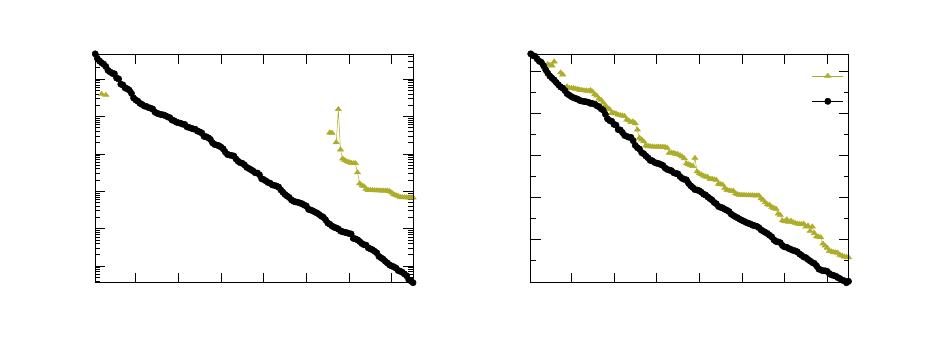}}%
    \gplfronttext
  \end{picture}%
\endgroup

  \caption{A comparison of the KL accuracy of approximating the joint
  distribution of $[\by_1, \by_2]$ using two different methods to produce low
  rank proxies for $\bS_{12}$. The first simply uses the truncated SVD to obtain
  $\bS_{12} \approx \bm{U} \bm{V}^T$, and the second implicitly builds it using
  the block-sparse inverse Cholesky factor via (\ref{eq:ichol_appx}). In the
  weakly separated regime, measurements for $\by_1$ and $\by_2$ may be
  arbitrarily close together, and in the strong regime they are at least $0.05$
  norm units apart.
  }
  \label{fig:odlr}
\end{figure}

As Figure \ref{fig:odlr} demonstrates, particularly in the weakly separated
regime, direct low rank approximation of $\bS_{12}$ does not yield positive
definite joint covariances for ranks less than $\approx 100$. And even then, its
accuracy is materially worse than the approach using this work via
(\ref{eq:ichol_appx}). In the strongly separated case, the two methods perform
more similarly, although for very low ranks the direct approximation still shows
instability. By forming $(\tilde{\bm{L}}^T \tilde{\bm{L}})^{-1}$, this approach
is also producing low rank approximations to $\bS_{12}$. But crucially, it is
not doing so with respect to a traditional error metric like the spectral norm.
Moreover, letting $\tilde{p}(\by_1, \by_2) = p(\by_1) p(\by_2 \sv \bm{Q}^T
\by_1)$ be the approximated joint density, we note that
\begin{equation*} 
 \kl{p(\by_1, \by_2)}{\tilde{p}(\by_1, \by_2)}
 =
 \E_{\by_1 \sim \Nd(\bm{0}, \bS_{11})} \kl{p(\by_2 \sv \by_1)}{p(\by_2 \sv
 \bm{Q}^T \by_1)}.
\end{equation*}
In this setting of rectangular $\mathcal{D}$ and $\Dpred$ and a kernel
that is analytic away from the origin, we see precisely the exponential
convergence implied by the theory of Section \ref{sec:conditioning}.

\subsection{Sea surface temperature anomaly prediction} \label{sec:ssta}

To provide a demonstration with real data in non-idealized settings, we consider
the problem of predicting an irregular gap that was synthetically created to
resemble cloud-based measurement gaps in sea surface temperature anomaly data in
the central Pacific ocean from the NOAA Coral Reef Watch database
\cite{noaa_crw}.  This selection of data has approximately $58\,000$ given
measurements and $9\,000$ missing points that need to be interpolated. The
simulated cloud-based gap does not have a smooth boundary or simple geometry,
and so designing specific orthogonal polynomials to obtain well-behaved
interpolation nodes is not feasible. Instead, we demonstrate another practical
but mathematically sub-optimal compromise of simply using the missing locations
themselves for $\set{\bx_l^g}$. After fitting several mean-zero GP models to the
data, the generalized Cauchy covariance function, given by
\begin{equation*} 
  K(\bx - \bx') = \sigma^2 \left(
    1 + \set{
      \frac{|x_{\text{lat}} - x'_{\text{lat}}|^2 }{\rho_\text{lat}}
      +
      \frac{|x_{\text{lon}} - x'_{\text{lon}}|^2 }{\rho_\text{lon}}
    }^{\alpha/2}
  \right)^{-\beta/\alpha},
\end{equation*}
provided the best fit and was selected for prediction. The specific parameter
estimates, obtained by optimizing a Vecchia-approximated log-likelihood, were
$(\sigma, \rho_{\text{lat}}, \rho_{\text{lon}}, \alpha, \beta) \approx (0.7,
0.33, 0.34, 1.97, 0.32)$, indicating a process with strong long memory. Since
these measurements are given on a grid, standard FFT-based acceleration methods
for matrix-vector products with $\bS$ can be used in conjunction with
Vecchia-based preconditioners to rapidly obtain solves $\bv \mapsto \bS^{-1}
\bv$ at the cost of approximately $\sim 10-20$ FFTs.  As a final detail, we
remove prediction points that are within a sufficiently small distance from
multiple data points that nearest neighbor prediction is obviously a better
choice and would be chosen in practice.
\begin{figure}[!ht]
  \centering
\begingroup
  \makeatletter
  \providecommand\color[2][]{%
    \GenericError{(gnuplot) \space\space\space\@spaces}{%
      Package color not loaded in conjunction with
      terminal option `colourtext'%
    }{See the gnuplot documentation for explanation.%
    }{Either use 'blacktext' in gnuplot or load the package
      color.sty in LaTeX.}%
    \renewcommand\color[2][]{}%
  }%
  \providecommand\includegraphics[2][]{%
    \GenericError{(gnuplot) \space\space\space\@spaces}{%
      Package graphicx or graphics not loaded%
    }{See the gnuplot documentation for explanation.%
    }{The gnuplot epslatex terminal needs graphicx.sty or graphics.sty.}%
    \renewcommand\includegraphics[2][]{}%
  }%
  \providecommand\rotatebox[2]{#2}%
  \@ifundefined{ifGPcolor}{%
    \newif\ifGPcolor
    \GPcolortrue
  }{}%
  \@ifundefined{ifGPblacktext}{%
    \newif\ifGPblacktext
    \GPblacktexttrue
  }{}%
  \let\gplgaddtomacro\g@addto@macro
  \gdef\gplbacktext{}%
  \gdef\gplfronttext{}%
  \makeatother
  \ifGPblacktext
    \def\colorrgb#1{}%
    \def\colorgray#1{}%
  \else
    \ifGPcolor
      \def\colorrgb#1{\color[rgb]{#1}}%
      \def\colorgray#1{\color[gray]{#1}}%
      \expandafter\def\csname LTw\endcsname{\color{white}}%
      \expandafter\def\csname LTb\endcsname{\color{black}}%
      \expandafter\def\csname LTa\endcsname{\color{black}}%
      \expandafter\def\csname LT0\endcsname{\color[rgb]{1,0,0}}%
      \expandafter\def\csname LT1\endcsname{\color[rgb]{0,1,0}}%
      \expandafter\def\csname LT2\endcsname{\color[rgb]{0,0,1}}%
      \expandafter\def\csname LT3\endcsname{\color[rgb]{1,0,1}}%
      \expandafter\def\csname LT4\endcsname{\color[rgb]{0,1,1}}%
      \expandafter\def\csname LT5\endcsname{\color[rgb]{1,1,0}}%
      \expandafter\def\csname LT6\endcsname{\color[rgb]{0,0,0}}%
      \expandafter\def\csname LT7\endcsname{\color[rgb]{1,0.3,0}}%
      \expandafter\def\csname LT8\endcsname{\color[rgb]{0.5,0.5,0.5}}%
    \else
      \def\colorrgb#1{\color{black}}%
      \def\colorgray#1{\color[gray]{#1}}%
      \expandafter\def\csname LTw\endcsname{\color{white}}%
      \expandafter\def\csname LTb\endcsname{\color{black}}%
      \expandafter\def\csname LTa\endcsname{\color{black}}%
      \expandafter\def\csname LT0\endcsname{\color{black}}%
      \expandafter\def\csname LT1\endcsname{\color{black}}%
      \expandafter\def\csname LT2\endcsname{\color{black}}%
      \expandafter\def\csname LT3\endcsname{\color{black}}%
      \expandafter\def\csname LT4\endcsname{\color{black}}%
      \expandafter\def\csname LT5\endcsname{\color{black}}%
      \expandafter\def\csname LT6\endcsname{\color{black}}%
      \expandafter\def\csname LT7\endcsname{\color{black}}%
      \expandafter\def\csname LT8\endcsname{\color{black}}%
    \fi
  \fi
    \setlength{\unitlength}{0.0500bp}%
    \ifx\gptboxheight\undefined%
      \newlength{\gptboxheight}%
      \newlength{\gptboxwidth}%
      \newsavebox{\gptboxtext}%
    \fi%
    \setlength{\fboxrule}{0.5pt}%
    \setlength{\fboxsep}{1pt}%
    \definecolor{tbcol}{rgb}{1,1,1}%
\begin{picture}(9060.00,2820.00)%
    \gplgaddtomacro\gplbacktext{%
      \csname LTb\endcsname
      \put(803,690){\makebox(0,0)[r]{\strut{}\tiny 12}}%
      \csname LTb\endcsname
      \put(803,971){\makebox(0,0)[r]{\strut{}\tiny 14}}%
      \csname LTb\endcsname
      \put(803,1252){\makebox(0,0)[r]{\strut{}\tiny 16}}%
      \csname LTb\endcsname
      \put(803,1533){\makebox(0,0)[r]{\strut{}\tiny 18}}%
      \csname LTb\endcsname
      \put(803,1814){\makebox(0,0)[r]{\strut{}\tiny 20}}%
      \csname LTb\endcsname
      \put(803,2095){\makebox(0,0)[r]{\strut{}\tiny 22}}%
      \csname LTb\endcsname
      \put(803,2375){\makebox(0,0)[r]{\strut{}\tiny 24}}%
      \csname LTb\endcsname
      \put(1034,360){\makebox(0,0){\strut{}\tiny -114}}%
      \csname LTb\endcsname
      \put(1542,360){\makebox(0,0){\strut{}\tiny -110}}%
      \csname LTb\endcsname
      \put(2050,360){\makebox(0,0){\strut{}\tiny -106}}%
      \csname LTb\endcsname
      \put(2557,360){\makebox(0,0){\strut{}\tiny -102}}%
    }%
    \gplgaddtomacro\gplfronttext{%
      \csname LTb\endcsname
      \put(406,1469){\rotatebox{-270.00}{\makebox(0,0){\strut{}\scriptsize Latitude}}}%
      \csname LTb\endcsname
      \put(1853,60){\makebox(0,0){\strut{}\scriptsize Longitude}}%
      \csname LTb\endcsname
      \put(1853,2759){\makebox(0,0){\strut{}\scriptsize SST Anomaly (deg. C)}}%
    }%
    \gplgaddtomacro\gplbacktext{%
      \csname LTb\endcsname
      \put(3470,678){\makebox(0,0)[r]{\strut{}\tiny 14}}%
      \csname LTb\endcsname
      \put(3470,974){\makebox(0,0)[r]{\strut{}\tiny 15}}%
      \csname LTb\endcsname
      \put(3470,1270){\makebox(0,0)[r]{\strut{}\tiny 16}}%
      \csname LTb\endcsname
      \put(3470,1566){\makebox(0,0)[r]{\strut{}\tiny 17}}%
      \csname LTb\endcsname
      \put(3470,1861){\makebox(0,0)[r]{\strut{}\tiny 18}}%
      \csname LTb\endcsname
      \put(3470,2157){\makebox(0,0)[r]{\strut{}\tiny 19}}%
      \csname LTb\endcsname
      \put(3470,2453){\makebox(0,0)[r]{\strut{}\tiny 20}}%
      \csname LTb\endcsname
      \put(3651,360){\makebox(0,0){\strut{}\tiny -112}}%
      \csname LTb\endcsname
      \put(4240,360){\makebox(0,0){\strut{}\tiny -110}}%
      \csname LTb\endcsname
      \put(4829,360){\makebox(0,0){\strut{}\tiny -108}}%
      \csname LTb\endcsname
      \put(5417,360){\makebox(0,0){\strut{}\tiny -106}}%
    }%
    \gplgaddtomacro\gplfronttext{%
      \csname LTb\endcsname
      \put(4519,60){\makebox(0,0){\strut{}\scriptsize Longitude}}%
      \csname LTb\endcsname
      \put(5712,677){\makebox(0,0)[l]{\strut{}\tiny -2}}%
      \csname LTb\endcsname
      \put(5712,1073){\makebox(0,0)[l]{\strut{}\tiny -1}}%
      \csname LTb\endcsname
      \put(5712,1469){\makebox(0,0)[l]{\strut{}\tiny 0}}%
      \csname LTb\endcsname
      \put(5712,1866){\makebox(0,0)[l]{\strut{}\tiny 1}}%
      \csname LTb\endcsname
      \put(5712,2262){\makebox(0,0)[l]{\strut{}\tiny 2}}%
      \csname LTb\endcsname
      \put(4519,2759){\makebox(0,0){\strut{}\scriptsize holdout data}}%
    }%
    \gplgaddtomacro\gplbacktext{%
      \csname LTb\endcsname
      \put(6136,678){\makebox(0,0)[r]{\strut{}}}%
      \csname LTb\endcsname
      \put(6136,974){\makebox(0,0)[r]{\strut{}}}%
      \csname LTb\endcsname
      \put(6136,1270){\makebox(0,0)[r]{\strut{}}}%
      \csname LTb\endcsname
      \put(6136,1566){\makebox(0,0)[r]{\strut{}}}%
      \csname LTb\endcsname
      \put(6136,1861){\makebox(0,0)[r]{\strut{}}}%
      \csname LTb\endcsname
      \put(6136,2157){\makebox(0,0)[r]{\strut{}}}%
      \csname LTb\endcsname
      \put(6136,2453){\makebox(0,0)[r]{\strut{}}}%
      \csname LTb\endcsname
      \put(6318,360){\makebox(0,0){\strut{}\tiny -112}}%
      \csname LTb\endcsname
      \put(6907,360){\makebox(0,0){\strut{}\tiny -110}}%
      \csname LTb\endcsname
      \put(7495,360){\makebox(0,0){\strut{}\tiny -108}}%
      \csname LTb\endcsname
      \put(8084,360){\makebox(0,0){\strut{}\tiny -106}}%
    }%
    \gplgaddtomacro\gplfronttext{%
      \csname LTb\endcsname
      \put(7186,60){\makebox(0,0){\strut{}\scriptsize Longitude}}%
      \csname LTb\endcsname
      \put(8379,420){\makebox(0,0)[l]{\strut{}\tiny -0.4}}%
      \csname LTb\endcsname
      \put(8379,610){\makebox(0,0)[l]{\strut{}\tiny -0.3}}%
      \csname LTb\endcsname
      \put(8379,801){\makebox(0,0)[l]{\strut{}\tiny -0.2}}%
      \csname LTb\endcsname
      \put(8379,992){\makebox(0,0)[l]{\strut{}\tiny -0.1}}%
      \csname LTb\endcsname
      \put(8379,1183){\makebox(0,0)[l]{\strut{}\tiny 0}}%
      \csname LTb\endcsname
      \put(8379,1374){\makebox(0,0)[l]{\strut{}\tiny 0.1}}%
      \csname LTb\endcsname
      \put(8379,1565){\makebox(0,0)[l]{\strut{}\tiny 0.2}}%
      \csname LTb\endcsname
      \put(8379,1756){\makebox(0,0)[l]{\strut{}\tiny 0.3}}%
      \csname LTb\endcsname
      \put(8379,1947){\makebox(0,0)[l]{\strut{}\tiny 0.4}}%
      \csname LTb\endcsname
      \put(8379,2138){\makebox(0,0)[l]{\strut{}\tiny 0.5}}%
      \csname LTb\endcsname
      \put(8379,2329){\makebox(0,0)[l]{\strut{}\tiny 0.6}}%
      \csname LTb\endcsname
      \put(8379,2519){\makebox(0,0)[l]{\strut{}\tiny 0.7}}%
      \csname LTb\endcsname
      \put(7186,2759){\makebox(0,0){\strut{}\scriptsize $\hat{\by}^{\text{missing}}_{\bm{Q}} - \hat{\by}^{\text{missing}}_{\text{Vecchia}}$}}%
    }%
    \gplbacktext
    \put(0,0){\includegraphics[width={453.00bp},height={141.00bp}]{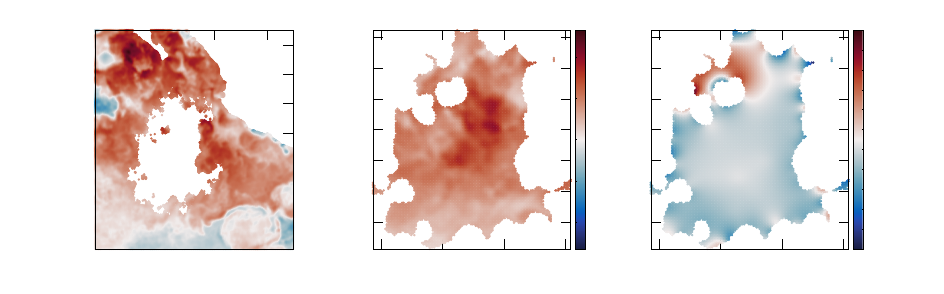}}%
    \gplfronttext
  \end{picture}%
\endgroup

  \caption{The SST anomaly data used for the prediction study. Left panel: the
  data treated as given. Middle panel: the holdout data from a synthetic cloud
  mask. Right panel: the difference in predictions between our method introduced
  here with $\bm{Q} \in \R^{n \times 35}$ and a Vecchia approximation with
  $k=200$ conditioning points for marginal data covariance and prediction.}
  \label{fig:ssta}
\end{figure}

Figure \ref{fig:ssta} shows the given and synthetically missing data in the left
two panels. To assess the quality of this new prediction method, we compare the
estimates and conditional variances one gets with those obtained from a Vecchia
approximation \cite{vecchia1988,stein2004,datta2016,finley2019,katzfuss2021},
which is arguably the clear best general-purpose approach for large-scale
estimation and prediction of Gaussian process models for environmental datasets
like this one \cite{heaton2019}. For the sake of fairness, we increased the
number of conditioning points (nearest neighbors) in the Vecchia approximation
until obtaining the conditional mean and variance of the prediction points took
longer than our approach. The rightmost panel of Figure \ref{fig:ssta} shows the
difference in these conditional means. Interestingly, we see that it is not
largest in the center of the domain where all measurements are most distant, but
instead in a region of the domain by a small pocket of proximal data. Table
\ref{tab:ssta} gives a numerical summary of prediction RMSE and continuous rank
probability score (CRPS) \cite{gneiting2007} for the two methods, indicating an
improvement in RMSE of approximately $14\%$ and in CRPS of $10\%$. Interestingly,
we note that even continuing to increase the conditioning set size much farther
in the Vecchia-based approach yields little improvement, potentially indicating
a strong non-locality in the prediction problem in that part of the missing domain.
\begin{table}[!ht]
\label{tab:ssta}
\centering
\begin{tabular}{cccc}
Method & runtime (s) & RMSE & CRPS \\
\hline
Our method ($\bm{Q} \in \R^{n \times 35}$) & $\bm{63}$  & $\bm{0.282}$ & $\bm{0.156}$ \\
Vecchia ($k=200$)                          & $78$       & $0.327$      & $0.174$ \\
Vecchia ($k=300$)                          & $173$      & $0.320$      & $0.170$
\end{tabular}
\caption{A comparison of the prediction method introduced here and a Vecchia
approximation with a high number of conditioning points for the problem of
predicting the cloud-based missing values in the SST anomaly data from Figure
\ref{fig:ssta}.} 
\end{table}

\subsection{Prediction runtime comparisons} \label{sec:pred_runtime}

As a final demonstration, we compare the cost of predicting at $m$ points in a
missing region using either $k$ nearest neighbor prediction or our method
described in this work. The setting of the prediction problem is similar to that
described in the runtime cost demonstration above: $n \approx 10\,000$ points on
$[0,1]^2$ are sampled and locations in $B_{0.2}([1/2, 1/2])$ are removed. Figure
\ref{fig:pred_times} gives a visual summary of the cost comparison of this
method versus $k$ nearest neighbor prediction. Naturally the pre-computation
overhead of this method is much higher than $k$ nearest neighbor conditioning,
as the only precomputation required for the latter method is assembling the K-D
tree, for which there are many exceptionally fast software libraries available.
After assembling $\bm{Q}$ and performing the anterpolation compressions,
however, we see that prediction is so fast that the runtime changes from $2$s
for $m=2^{10} \approx 1\,000$ to only $4$s for $m=2^{15} \approx 30\,000$. 
\begin{figure}
  \centering
\begingroup
  \makeatletter
  \providecommand\color[2][]{%
    \GenericError{(gnuplot) \space\space\space\@spaces}{%
      Package color not loaded in conjunction with
      terminal option `colourtext'%
    }{See the gnuplot documentation for explanation.%
    }{Either use 'blacktext' in gnuplot or load the package
      color.sty in LaTeX.}%
    \renewcommand\color[2][]{}%
  }%
  \providecommand\includegraphics[2][]{%
    \GenericError{(gnuplot) \space\space\space\@spaces}{%
      Package graphicx or graphics not loaded%
    }{See the gnuplot documentation for explanation.%
    }{The gnuplot epslatex terminal needs graphicx.sty or graphics.sty.}%
    \renewcommand\includegraphics[2][]{}%
  }%
  \providecommand\rotatebox[2]{#2}%
  \@ifundefined{ifGPcolor}{%
    \newif\ifGPcolor
    \GPcolortrue
  }{}%
  \@ifundefined{ifGPblacktext}{%
    \newif\ifGPblacktext
    \GPblacktexttrue
  }{}%
  \let\gplgaddtomacro\g@addto@macro
  \gdef\gplbacktext{}%
  \gdef\gplfronttext{}%
  \makeatother
  \ifGPblacktext
    \def\colorrgb#1{}%
    \def\colorgray#1{}%
  \else
    \ifGPcolor
      \def\colorrgb#1{\color[rgb]{#1}}%
      \def\colorgray#1{\color[gray]{#1}}%
      \expandafter\def\csname LTw\endcsname{\color{white}}%
      \expandafter\def\csname LTb\endcsname{\color{black}}%
      \expandafter\def\csname LTa\endcsname{\color{black}}%
      \expandafter\def\csname LT0\endcsname{\color[rgb]{1,0,0}}%
      \expandafter\def\csname LT1\endcsname{\color[rgb]{0,1,0}}%
      \expandafter\def\csname LT2\endcsname{\color[rgb]{0,0,1}}%
      \expandafter\def\csname LT3\endcsname{\color[rgb]{1,0,1}}%
      \expandafter\def\csname LT4\endcsname{\color[rgb]{0,1,1}}%
      \expandafter\def\csname LT5\endcsname{\color[rgb]{1,1,0}}%
      \expandafter\def\csname LT6\endcsname{\color[rgb]{0,0,0}}%
      \expandafter\def\csname LT7\endcsname{\color[rgb]{1,0.3,0}}%
      \expandafter\def\csname LT8\endcsname{\color[rgb]{0.5,0.5,0.5}}%
    \else
      \def\colorrgb#1{\color{black}}%
      \def\colorgray#1{\color[gray]{#1}}%
      \expandafter\def\csname LTw\endcsname{\color{white}}%
      \expandafter\def\csname LTb\endcsname{\color{black}}%
      \expandafter\def\csname LTa\endcsname{\color{black}}%
      \expandafter\def\csname LT0\endcsname{\color{black}}%
      \expandafter\def\csname LT1\endcsname{\color{black}}%
      \expandafter\def\csname LT2\endcsname{\color{black}}%
      \expandafter\def\csname LT3\endcsname{\color{black}}%
      \expandafter\def\csname LT4\endcsname{\color{black}}%
      \expandafter\def\csname LT5\endcsname{\color{black}}%
      \expandafter\def\csname LT6\endcsname{\color{black}}%
      \expandafter\def\csname LT7\endcsname{\color{black}}%
      \expandafter\def\csname LT8\endcsname{\color{black}}%
    \fi
  \fi
    \setlength{\unitlength}{0.0500bp}%
    \ifx\gptboxheight\undefined%
      \newlength{\gptboxheight}%
      \newlength{\gptboxwidth}%
      \newsavebox{\gptboxtext}%
    \fi%
    \setlength{\fboxrule}{0.5pt}%
    \setlength{\fboxsep}{1pt}%
    \definecolor{tbcol}{rgb}{1,1,1}%
\begin{picture}(7920.00,2820.00)%
    \gplgaddtomacro\gplbacktext{%
      \csname LTb\endcsname
      \put(1084,866){\makebox(0,0)[r]{\strut{}\footnotesize $0.1$}}%
      \csname LTb\endcsname
      \put(1084,1417){\makebox(0,0)[r]{\strut{}\footnotesize $1$}}%
      \csname LTb\endcsname
      \put(1084,1749){\makebox(0,0)[r]{\strut{}\footnotesize $4$}}%
      \csname LTb\endcsname
      \put(1084,1968){\makebox(0,0)[r]{\strut{}\footnotesize $10$}}%
      \csname LTb\endcsname
      \put(1084,2300){\makebox(0,0)[r]{\strut{}\footnotesize $40$}}%
      \csname LTb\endcsname
      \put(1084,2466){\makebox(0,0)[r]{\strut{}\footnotesize $80$}}%
      \csname LTb\endcsname
      \put(1185,460){\makebox(0,0){\strut{}\footnotesize $2^{10}$}}%
      \csname LTb\endcsname
      \put(2212,460){\makebox(0,0){\strut{}\footnotesize $2^{11}$}}%
      \csname LTb\endcsname
      \put(3239,460){\makebox(0,0){\strut{}\footnotesize $2^{12}$}}%
      \csname LTb\endcsname
      \put(4265,460){\makebox(0,0){\strut{}\footnotesize $2^{13}$}}%
      \csname LTb\endcsname
      \put(5292,460){\makebox(0,0){\strut{}\footnotesize $2^{14}$}}%
      \csname LTb\endcsname
      \put(6319,460){\makebox(0,0){\strut{}\footnotesize $2^{15}$}}%
    }%
    \gplgaddtomacro\gplfronttext{%
      \csname LTb\endcsname
      \put(7350,2400){\makebox(0,0)[r]{\strut{}\footnotesize $\bm{Q} \in \R^{n \times 100}$}}%
      \csname LTb\endcsname
      \put(7350,2160){\makebox(0,0)[r]{\strut{}\footnotesize $k=100$}}%
      \csname LTb\endcsname
      \put(7350,1920){\makebox(0,0)[r]{\strut{}\footnotesize $k=200$}}%
      \csname LTb\endcsname
      \put(7350,1680){\makebox(0,0)[r]{\strut{}\footnotesize $k=400$}}%
      \csname LTb\endcsname
      \put(486,1609){\rotatebox{-270.00}{\makebox(0,0){\strut{}\small runtime (s)}}}%
      \csname LTb\endcsname
      \put(3752,160){\makebox(0,0){\strut{}\small number of prediction points $m$}}%
    }%
    \gplbacktext
    \put(0,0){\includegraphics[width={396.00bp},height={141.00bp}]{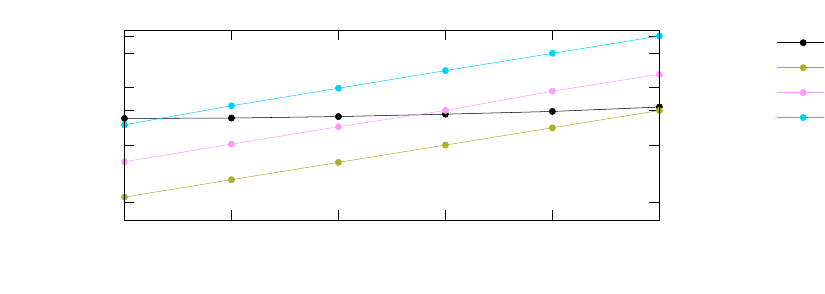}}%
    \gplfronttext
  \end{picture}%
\endgroup

  \caption{A demonstration of the precomputation-inclusive runtime cost for
  predicting at $m$ locations in $B_{0.2}([1/2, 1/2])$ with $n \approx 10\,000$
  locations in $[0,1]^2 \setminus B_{0.2}([1/2, 1/2])$. The black line gives the
  runtime cost of our method, including the precomputation cost of assembling
  $\bm{Q}$ and performing the anterpolation compressions to relative tolerance
  $10^{-9}$. 
  }
  \label{fig:pred_times}
\end{figure}
The number of conditioning points used for nearest neighbor
conditioning in Figure \ref{fig:pred_times} also demonstrates break-even costs:
if predicting with $k=100$ nearest neighbors gives the desired level of
accuracy, then this method becomes competitive for this problem if one
predicts at $m=2^{15}$ points or more. But if larger numbers of conditioning
points are necessary, the inflection point at which this approach becomes both
faster \emph{and} likely more accurate comes much sooner: at barely more than
$1\,000$ prediction points, the precomputation burden of this method is made up
for if $k=500$ conditioning points are needed. And based on the results of the
last section, circumstances in which even more than $k=500$ points are required
are not restricted to pathological examples.

\section{Discussion} \label{sec:discussion}

We have demonstrated and analyzed an approach for achieving high-accuracy
Gaussian process predictions in regions of unobserved data that is complementary
to other fast online-compatible prediction metrics like $k$-nearest neighbors.
The key observation is the design of specific linear combinations that capture
the relevant variation in the given data as it pertains to the prediction
problem. Strategies for rapidly obtaining and compressing these precomputations
to the point where $\bO(1)$-cost predictions are given, making this approach
particularly attractive for settings in which one wants to predict at a
significant number of locations that are not a priori known. 

A variety of questions on this topic are potentially exciting areas of future
work. A natural extension to study would be judiciously combining $k$-nearest
neighbor prediction and dense $\bm{Q}^T \by$ conditioning as introduced here for
sufficiently far points, and perhaps doing so dynamically.  Another particularly
exciting option is to formulate and analyze a hierarchical version of the
approximation introduced in Section \ref{sec:odlr}, which may admit sharp
efficiency bounds and give a generally useful fast algorithm for something like
an inverse square-root FMM. Finally, for both of these future projects and many
more applications, it would also be very valuable to have better methods for
selecting efficient interpolation points on complex domains.  Tooling from the
theory and computational techniques of numerical quadrature may prove
particularly useful.

\bibliography{combined}

\end{document}